\documentclass[11pt]{amsart}
\usepackage[foot]{amsaddr}
\usepackage{hyperref}
\usepackage[margin=1.2in]{geometry}
\usepackage{amsmath,amsfonts,amsthm, amssymb}
\usepackage{graphicx}
\usepackage[all]{xy}
\usepackage{mathrsfs}
\usepackage{enumerate}
\usepackage{braket}
\usepackage[dvipsnames]{xcolor}
\usepackage[capitalise]{cleveref}  

\hypersetup{breaklinks=true,allcolors=blue,colorlinks=true}

\colorlet{mylinkcolor}{violet}
\colorlet{mycitecolor}{YellowOrange}
\colorlet{myurlcolor}{Aquamarine}

\theoremstyle{plain}

\newtheorem{theorem}{Theorem}[section]

\newtheorem{proposition}[theorem]{Proposition}
\newtheorem{lemma}[theorem]{Lemma}

\theoremstyle{definition}
\newtheorem{definition}[theorem]{Definition}

\theoremstyle{remark}
\newtheorem{remark}[theorem]{Remark}

\DeclareMathOperator{\polylog}{polylog}
\DeclareMathOperator{\Ker}{Ker}
\DeclareMathOperator{\id}{id}
\DeclareMathOperator{\Hom}{Hom}

\DeclareMathOperator{\End}{End}
\DeclareMathOperator{\Ell}{Ell}
\DeclareMathOperator{\Cl}{Cl}
\DeclareMathOperator{\Ab}{Ab}
\DeclareMathOperator{\Mod}{Mod}
\DeclareMathOperator{\isom}{isom}
\DeclareMathOperator{\Mumf}{Mumf}
\DeclareMathOperator{\Mat}{Mat}

\newcommand{\F}{\mathbb{F}}
\newcommand{\Z}{\mathbb{Z}}
\newcommand{\C}{\mathbb{C}}
\newcommand{\Q}{\mathbb{Q}}
\newcommand{\pro}{\mathbb{P}}

\newcommand{\adv}{\mathcal{A}}

\newcommand{\scrO}{\mathscr{O}}
\newcommand{\leg}[2]{\left(\frac{#1}{#2}\right)}
\newcommand{\deq}{\mathrel{\mathop:}=}

\newcommand{\mapgen}{{\normalfont\textsf{MapGen}}}
\newcommand{\Setup}{{\normalfont\textsf{Setup}}}

\newcommand{\pp}{{\normalfont\textsf{pp}}}

\newcommand{\NP}{{\normalfont\textsf{NP}}}

\newcommand{\Constrain}{{\normalfont\textsf{Constrain}}}
\newcommand{\Eval}{{\normalfont\textsf{Eval}}}

\newcommand{\Enc}{{\normalfont\textsf{Enc}}}
\newcommand{\Dec}{{\normalfont\textsf{Dec}}}

\newcommand{\rgets}{\mathrel{\mathpalette\rgetscmd\relax}}
\newcommand{\rgetscmd}{\ooalign{$\leftarrow$\cr
    \hidewidth\raisebox{1.2\height}{\scalebox{0.5}{\ \rm R}}\hidewidth\cr}}
\newcommand{\defn}[1]{\textbf{#1}}

\title[Multiparty NIKE and More From Isogenies]{%
Multiparty Non-Interactive Key Exchange\\ and More From
Isogenies on Elliptic Curves}

\author[Boneh]{Dan Boneh$^1$}
\address{$^1$Stanford University}
\email{dabo@cs.stanford.edu}

\author[Glass]{Darren Glass$^2$}
\address{$^2$Gettysburg College}
\email{dglass@gettysburg.edu}

\author[Krashen]{Daniel Krashen$^3$}
\address{$^3$University of Georgia, Athens}
\email{daniel.krashen@gmail.com}

\author[Lauter]{Kristin Lauter$^4$}
\address{$^4$Microsoft Research}
\email{klauter@microsoft.com}

\author[Sharif]{Shahed Sharif$^5$}
\address{$^5$California State University San Marcos}
\email{ssharif@csusm.edu}

\author[Silverberg]{Alice Silverberg$^6$}
\address{$^6$University of California, Irvine}
\email{asilverb@uci.edu}

\author[Tibouchi]{Mehdi Tibouchi$^7$}
\address{$^7$NTT Corporation}
\email{tibouchi.mehdi@lab.ntt.co.jp}

\author[Zhandry]{Mark Zhandry$^8$}
\address{$^8$Princeton University}
\email{mzhandry@princeton.edu}

\newcommand{\ssparagraph}[1]{\subsection*{#1}}
\newcommand{\sssparagraph}[1]{\subsubsection*{#1}}
\date{}

\newcommand{\paperabstract}{%
We describe a framework for constructing an efficient non-interactive key
exchange (NIKE) protocol for $n$ parties for any $n \geq 2$.  Our approach is
based on the problem of computing isogenies between isogenous elliptic
curves, which is believed to be difficult. We do not obtain a working
protocol because of a missing step that is currently an open mathematical problem.
What we need to complete our protocol is an efficient algorithm that
takes as input an abelian variety presented as a product of isogenous
elliptic curves, and outputs an isomorphism invariant of the abelian
variety.

Our framework builds a {\em cryptographic invariant map}, which is a new
primitive closely related to a cryptographic multilinear map, but whose
range does not necessarily have a  group structure. Nevertheless, we show
that a cryptographic invariant map can be used to build several
cryptographic primitives, including NIKE, that were previously 
constructed from multilinear maps and indistinguishability obfuscation.}
\newcommand{\paperkeywords}{%
Multilinear maps, Non-Interactive Key Exchange, Isogenies,
Witness Encryption, Abelian Varieties}

\begin{document}
\begin{abstract}\paperabstract\end{abstract}

\maketitle

\section{Introduction}

Let $\F_q$ be a finite field, let $E$ be an ordinary elliptic curve over
$\F_q$, and let $X$ be the set of elliptic curves over $\F_q$ that are isogenous
to $E$. The set $X$ is almost always large (containing on the order of
$\sqrt{q}$ elements). Moreover, under suitable conditions on $E$, the set $X$ is
endowed with a free and transitive action $\ast$ by a certain abelian
group $G$, which is the ideal class group of the endomorphism
ring of $E$. The action $\ast$ maps a given $g \in G$ and $E \in X$ to a curve
$g \ast E \in X$.  

This action, originally defined by Deuring~\cite{Deuring}, has a number of properties that makes it useful in cryptography.
First, for a fixed curve $E \in X$, the map $G \to X$ defined by $g \mapsto g \ast E$ is believed to be a one-way function.
In other words, given a random curve $E' \in X$ it is difficult to find an element $g \in G$ 
such that $E' = g \ast E$.  This suggests a
Diffie--Hellman two-party key exchange protocol, proposed by
Couveignes~\cite{EPRINT:Couveignes06}
and Rostovtsev and Stolbunov~\cite{EPRINT:RosSto06}: Alice chooses a random $a \in G$
and publishes $E_a \deq a \ast E$;  Bob chooses a random $b \in G$ and publishes $E_b \deq b \ast E$.
Their shared key is the curve 
$$E_{ab} \deq (ab) \ast E = a \ast E_b = b \ast E_a,$$
which they can both compute.
To ensure that both parties obtain the same key, their shared key is the $j$-invariant of the curve $E_{ab}$.
More recently, De Feo, Jao, and Pl{\^u}t~\cite{FJP14}, Galbraith~\cite{Galbraith}, Castryck {et al.}~\cite{CSIDH},
and De Feo, Kieffer, and Smith~\cite{cryptoeprint:2018:485} proposed
variants of this protocol with better security and efficiency.  
Moreover, a supersingular version of the isogeny problem was introduced and proposed as the basis for a collision resistant hash function~\cite{JC:ChaLauGor09}. 
Security of this one-way function was further studied in~\cite{EC:EHLMP18}.

Second, as alluded to above, the star operator satisfies the following useful property:
for all $g_1,\ldots,g_n \in G$ the abelian varieties
\[
A_1 \deq (g_1 \ast E) \times \cdots \times (g_n \ast E) \quad\text{ and }\quad A_2 \deq (g_1\cdots g_n) \ast E \times E^{n-1}
\]
are isomorphic (see Appendix~\ref{sec:lowbrow}).
As we will see in the next section, this suggests an $n$-party non-interactive key
exchange protocol, as well as many other cryptographic constructions. 
This property leads to a more general cryptographic primitive that we call a \defn{cryptographic invariant map},
defined in the next section.   
This primitive has properties that are similar to those of cryptographic
multilinear maps~\cite{boneh2003applications,EC:GarGenHal13},
which have found numerous applications in cryptography
(e.g,~\cite{FOCS:GGHRSW13,STOC:GGSW13,AC:BonWat13,EC:BLRSZZ15}). 
We discuss applications of cryptographic invariant maps in Section~\ref{sec:applications}.
In Remark~\ref{rem:supersingrmk} we explain why we use ordinary and not supersingular elliptic curves.
Section~\ref{Sec:IsogMaps} describes our approach to constructing
cryptographic invariant maps from isogenies.

This work leads to the following question in algebraic geometry.

\ssparagraph{An open problem}
To make the cryptographic applications discussed above viable we must first overcome
an important technical challenge.  While the varieties $A_1$ and $A_2$ defined above are isomorphic,
they are presented differently.  Our applications require an efficient way to compute an invariant
that is the same for $A_1$ and $A_2$.
In addition, the invariant must distinguish non-isomorphic varieties.  
We do not know any such computable isomorphism invariant, and we present this as an open problem.
In Section~\ref{sec:failure-theta-null} we explain why some natural
proposals for isomorphism invariants do not seem to work.  
In Remarks~\ref{rem:DDH} and~\ref{rem:isogDDH} we show that a solution to
this open problem, even for $n=2$, would solve the isogeny decision
Diffie--Hellman problem.  
Further, we give evidence that computing a particular isomorphism invariant might be equivalent
to solving the elliptic curve isogeny problem, which is believed (or hoped)
to be a
quantum-resistant hard problem. Thus, Section~\ref{sec:failure-theta-null} might be useful from
the point of view of cryptanalysis of isogeny-based cryptography.

\section{Cryptographic invariant maps}
\label{sec:defs}

\begin{definition}\label{Def:EFTAction}
Let $X$ be a finite set and let $G$ be a finite abelian group. 
We say that \defn{$G$ acts efficiently on $X$ freely and transitively} if there is an efficiently
computable map $\ast: G \times X \to X$ such that:
\begin{itemize}
\item the map is a group action: $g \ast (h \ast x) = (gh) \ast x$, and there is an identity element $\id \in G$ such that $\id \ast x = x$, for all $x \in X$ and all $g,h \in G$;
\item the action is transitive: for every $(x,y) \in X\times X$ there is a $g \in G$ such that $g \ast x = y$; and
\item the action is free: if $x \in X$ and $g,h \in G$ satisfy $g \ast x = h \ast x$, then $g=h$.
\end{itemize}
\end{definition}

\begin{definition}\label{Def:CrInvMap}
By a \defn{cryptographic invariant map} we mean a randomized 
algorithm $\mapgen$ that inputs a 
security parameter $\lambda$,
outputs public parameters $\pp = (X,S,G,e)$, and runs in time polynomial in $\lambda$, where: 
\begin{itemize}
\item $X$ and $S$ are sets, and $X$ is finite, 
\item $G$ is a finite abelian group that acts efficiently on $X$ freely and transitively, 
\item $e$ is a deterministic algorithm that runs in time polynomial in $\lambda$ and $n$, such that for each $n>0$, algorithm $e$ takes $\lambda$ as input and
computes a map $e_n:X^n \to S$ that satisfies:
\begin{itemize}
\item \defn{Invariance property} of $e_n$:  for all $x \in X$ and $g_1,\ldots,g_n \in G$, 
\[   e_n(g_1 \ast x, \ldots, g_n \ast x) = 
                   e_n\big((g_1 \cdots g_n) \ast x, x, \ldots, x\big);  \]
\item \defn{Non-degeneracy} of $e_n$:  
for all $i$ with $1 \le i \le n$ and  $x_1,\ldots,x_{i-1},x_{i+1},\ldots,x_n \in X$, the map $X \to S$ defined by 
$y \mapsto e_n(x_1,\ldots,x_{i-1},y,x_{i+1},\ldots,x_n)$ is injective.
\end{itemize}
\end{itemize}
\end{definition}

In our candidate instantiation for cryptographic invariant maps the set $X$ is a set of isogenous elliptic curves and
the group $G$ acting on $X$ is a class group.
The elements of $S$ are isomorphism invariants of products of elliptic curves.   

Definition \ref{Def:CrInvMap} is quite ambitious in that it asks that $e_n$ be defined for all $n>0$ and run in polynomial time in $n$ (and $\lambda$).  
A cryptographic invariant map that is defined even for a single $n>2$, and satisfies the security assumptions in the next subsection, would still be quite interesting.  
We require a construction that works for all $n$ because our framework using elliptic curve isogenies seems to support it.
Similarly, we note that a construction that works for all $n>0$, but runs in time exponential in~$n$ is still useful. 
It would limit our ability to evaluate $e_n$ to relatively small $n$, but that is still of great interest.
In the first three proposals in Section~\ref{sec:failure-theta-null} we study candidates for $e_n$ that run in time exponential in $n$,
satisfy the non-degeneracy property, but do not satisfy the invariance property.  
It is an open problem to find a map that also satisfies the invariance property.

\ssparagraph{Security assumptions}

Next, we define some security assumptions on cryptographic invariant maps.
The notation $x \rgets X$ will denote an independent uniform random variable $x$ over the set $X$.
Similarly, we use $x' \rgets A(y)$ to define a random variable $x'$ that is the
output of a randomized algorithm $A$ on input $y$. 

The $n$-way computational Diffie--Hellman assumption states that, given only the public parameters and $(g_1 \ast x, \ldots, g_n \ast x) \in X^n$, it is difficult to compute
$e_{n-1}\big((g_1 \cdots g_n) \ast x, x, \ldots, x\big)$. 
A precise definition is the following:

\begin{definition}
We say that $\mapgen$ satisfies the $n$-way 
\defn{computational Diffie--Hellman assumption ($n$-CDH)}  
if for every polynomial time algorithm $\adv$,
\[  \Pr\Big[ \adv(\pp,\ g_1 \ast x, \ldots, g_n \ast x) = 
                e_{n-1}\big((g_1 \cdots g_n) \ast x, x, \ldots, x\big) \Big]  \]
is a negligible function of $\lambda$,
when $\pp \rgets \mapgen(\lambda)$, $g_1,\ldots,g_n \rgets G$, and $x \rgets X$. 
\end{definition}

\begin{remark}
\label{rem:DDH}
The natural  $n$-way decision Diffie--Hellman assumption on $X$ does not hold when invariant maps exist. 
That is, for all $n>0$ it is easy to distinguish $(g_1 \cdots g_n) \ast x \in X$ 
from a random element of $X$, given only $x,g_1 \ast x, \ldots, g_n \ast x$.  
Given a challenge $y \in X$, simply check if 
$$e_n(y,x,\ldots,x) = e_n(g_1 \ast x, \ldots, g_n \ast x).$$
Equality holds if and only if $y = (g_1 \cdots g_n) \ast x$. 
However, in Definition \ref{nDDHdefn} we define an $n$-way decision Diffie--Hellman assumption for $e_{n-1}$. 
It states that it is hard to distinguish 
$e_{n-1}\big((g_1 \cdots g_n) \ast x, x, \ldots, x\big)$ from a random element in the image of $e_{n-1}$, given only the public parameters, $x$, and $(g_1 \ast x, \ldots, g_n \ast x) \in X^n$.
\end{remark}

\begin{definition}
\label{nDDHdefn}
We say that $\mapgen$ satisfies the $n$-way 
\defn{decision Diffie--Hellman assumption ($n$-DDH)}  
if the following two distributions, $\mathcal{P}_0$ and $\mathcal{P}_1$, are polynomially indistinguishable,
when $\pp \rgets \mapgen(\lambda)$,  $g_1,\ldots,g_n \rgets G$, and $x \rgets X$:
\begin{itemize}
\item 
$\mathcal{P}_0$ is 
$(\pp,\ g_1 \ast x, \ldots, g_n \ast x,\ s_0)$
where $s_0 = e_{n-1}\big((g_1 \cdots g_n) \ast x, x, \ldots, x\big)$, 

\item 
$\mathcal{P}_1$ is 
$(\pp,\ g_1 \ast x, \ldots, g_n \ast x,\ s_1)$
where $s_1$ is random in 
$\text{Im}(e_{n-1}) \subseteq S$.
\end{itemize}
\end{definition}

\section{Applications}
\label{sec:applications}

We show that suitable cryptographic invariant maps can be used to solve a number of important problems in cryptography.

\ssparagraph{$n$-way Non-Interactive Key Exchange (NIKE)}
We show how to use a cryptographic invariant map to construct a Non-Interactive Key Exchange (NIKE) protocol in which $n$ parties create a shared secret key that only they can efficiently calculate, without any interaction among the $n$ parties.
Currently, secure $n$-party NIKE for $n>3$ is only known from general purpose indistinguishability obfuscation (e.g.,~\cite{boneh2017multiparty}).
Our NIKE construction is similar to the one
in~\cite{JC:Joux04,boneh2003applications,EC:GarGenHal13} and satisfies a ``static'' notion of security.
\begin{itemize}
\item $\text{Setup}(\lambda)$:  run $(X,S,G,e) \rgets \mapgen(\lambda)$ and choose $x \rgets X$.
Output $\pp \deq (X,S,G,e,x)$. 

\item For $i=1,\ldots,n$, party~$i$ chooses a random $g_i \rgets G$, computes $x_i \deq g_i \ast x \in X$, and publishes $x_i$ on a public bulletin board.

\item The shared key between the $n$-parties is 
$$k \deq e_{n-1}\big((g_1 \cdots g_n) \ast x, x, \ldots, x\big) \in S.$$
Party $i\in\{ 1,\ldots,n\}$ computes $k$ by obtaining $x_1,\ldots,x_n$ from the bulletin board,
then choosing some $j \in \{1,\ldots,n\}$ where $j \neq i$,
and computing $$k = e_{n-1}(x_1,\ldots,x_{j-1},\ g_i \ast x_j,\ x_{j+1},\ldots,x_n) \in S,$$
where $x_i$ is omitted from the input to $e_{n-1}$.
\end{itemize}

All $n$ parties obtain the same key $k$ by the invariance property of $e_{n-1}$.
Static security 
follows from the $n$-way decision
Diffie--Hellman assumption, as in~\cite{boneh2003applications}.   
Alternatively, we can rely on the weaker $n$-way computational Diffie--Hellman assumption by 
applying a hash function $H:S \to K$ to the 
key $k$.  We model $H$ as a random oracle in the security analysis.
We leave the question of an adaptively-secure NIKE, in the sense
of~\cite{PKC:FHKP13,EPRINT:Rao14}, from an invariant map for future work.

\ssparagraph{Unique signatures and verifiable random functions (VRF)}
A digital signature scheme is made up of three algorithms: a key generation algorithm that outputs a
public key and a secret key, a signing algorithm that signs a given message using the secret key, and
a verification algorithm that verifies a signature on a given message using the public key.
A signature scheme is a \defn{unique signature scheme} if for every public key and every message,
there is at most one signature that will be accepted as a valid signature for that message under the public key.
While a number of unique signature schemes are known in the random oracle
model (e.g.,~\cite{EC:BelRog96,AC:BonLynSha01}), it is quite hard to 
construct unique signatures without random
oracles~\cite{C:Lysyanskaya02,PKC:DodYam05}. 
Unique signatures are closely related to a simpler object called a
verifiable random function, or VRF~\cite{FOCS:MicRabVad99}.
Previous results show how to construct unique signatures and VRFs from multilinear maps
without random oracles~\cite{boneh2003applications}.
The same constructions work with a cryptographic invariant map.
The unique signature scheme works as follows:
The secret key is a random $(g_{1,0},g_{1,1},\ldots,g_{n,0},g_{n,1}) \rgets G^{2n}$.
The public key is $(x,y_{1,0},\ldots,y_{n,1}) \in X^{2n+1}$ where $x \rgets X$ and $y_{i,b} \deq g_{i,b} \ast x$ for $i=1,\ldots,n$ and $b=0,1$.
The signature on an $n$-bit message $m \in \{0,1\}^n$ is $\sigma \deq (\prod_{i=1}^n g_{i,m_i}) \ast x \in X$.  
To verify a signature~$\sigma$, check that $e_n(\sigma, x,\ldots,x) = e_n\big(y_{1,m_1}, \ldots, y_{n,m_n}\big)$.
The security analysis of this construction is the same as in~\cite{boneh2003applications}.

\ssparagraph{Constrained PRFs and broadcast encryption} We next describe
how to construct \emph{constrained pseudorandom
functions}~\cite{AC:BonWat13,CCS:KPTZ13,PKC:BoyGolIva14} for
\emph{bit-fixing constraints} from a cryptographic invariant map.  Such
constrained PRFs in turn can be used to build broadcast encryption with
short ciphertexts~\cite{AC:BonWat13}.

A pseudorandom function (PRF) is a function $F:\mathcal{K}\times\mathcal{A}\rightarrow\mathcal{B}$ that is computable in polynomial time.  Here, $\mathcal{K}$ is the key space, $\mathcal{A}$ is the domain, and $\mathcal{B}$ is the codomain.  Intuitively, PRF security requires that, for a random key $k \in \mathcal{K}$, an adversary who obtains pairs $\big(a,\ F(k,a)\big)$, for $a \in \mathcal{A}$ of its choice, cannot distinguish these pairs from pairs $\big(a, f(a)\big)$ where $f$ is a random function $\mathcal{A} \to \mathcal{B}$.

A \emph{bit-fixing constrained PRF} is a PRF where a key $k \in \mathcal{K}$
can be constrained to only evaluate the PRF on a subset of the domain $\mathcal{A}$, where $\mathcal{A} = \{0,1\}^n$. 
Specifically, for $V \subseteq [n] = \{1,\ldots,n\}$ and a function $v: V \to \{0,1\}$, let $\mathcal{A}_v = \{a\in\mathcal{A}:\forall i \in V, a_i=v(i) \}$.
A constrained key $k_v$ enables one to evaluate $F(k,a)$ for all $a \in \mathcal{A}_v$, but reveals nothing about $F(k,a)$ for $a \notin \mathcal{A}_v$.
We refer to~\cite{AC:BonWat13} for the complete definition of this concept, 
and its many applications.

We now explain how to construct bit-fixing constrained PRFs from
cryptographic invariant maps.  The construction and security proof are
essentially the same as in Boneh and Waters~\cite{AC:BonWat13}, but translated to our setting.  One complication is that the construction of Boneh and Waters requires a way to operate on invariants in $S$.  We get around this by delaying the evaluation of the invariant to the very last step. We thus obtain the following bit-fixing constrained PRF:

\begin{itemize}
\item $\Setup(\lambda)$: 
run $(X,S,G,e)\rgets\mapgen(\lambda)$ and choose $x\rgets X$.  \\
Next choose $\alpha\rgets G$ and $d_{i,b}\rgets G$ for $i\in[n]$ and $b\in\{0,1\}$.  \\
Output the key $k=(X,S,G,e,\alpha,\{d_{i,b}\}_{i,b})$.
    
\item The PRF is defined as: 
$F(k,a)=e_n\big(\;(\alpha\times\prod_{i=1}^n d_{i,a_i})\ast x,\ x,\ \dots,\  x\big)$.   \\
Here, $a \in \{0,1\}^n$ specifies a subset product of the set of $d_{i,b}$'s.
    
\item $\Constrain(k,v)$: Let $V \subseteq [n]$ be the support of the function $v$, and assume $V$ is not empty. 
The constrained key $k_v$ is constructed as follows. Set $D_{i,b}=d_{i,b}\ast x$ for $i\notin V$.  
Let $i_0$ be the smallest element of $V$.  Choose $|V|-1$ random $g_i \in G$ for $i\in V\setminus \{i_0\}$, and set $g_{i_0}=\alpha\times\prod_{i\in V} d_{i,v_i}\times(\prod_{i\in V\setminus \{i_0\}} g_i)^{-1} \in G$.
Let $h_i = g_i\ast x$ for $i\in V$.   \\
The constrained key is $k_v=\big(\{D_{i,b}\}_{i\notin V,b\in\{0,1\}},\ \{h_i\}_{i\in V}\big)$.
    
\item $\Eval(k_v,a)$: To evaluate $F(k,a)$ using the constrained key $k_v$ do the following.
If $a\notin \mathcal{A}_v$, output $\diamond$.  Otherwise, for $i=1,\ldots,n,$ let 
$C_i= D_{i,a_i}$ if $i\notin V$, and let $C_i= h_i$ otherwise.
Output  $e_n(C_1,\dots,C_n)$.    
Then, by construction,
\[ e_n(C_1,\dots,C_n) = e_n\left(\left(\prod_{i\notin V} d_{i,a_i}\prod_{i\in V}g_i\right)\ast x,x,\dots,x\right) = F(k,a), \]
as required.
\end{itemize}

The security proof for this construction is as in~\cite{AC:BonWat13}.
This construction can be further extended to a verifiable random function
(VRF) by adapting Fuchsbauer~\cite{SCN:Fuchsbauer14} in a similar fashion.

\ssparagraph{Witness encryption}  Witness encryption, due to Garg et
al.~\cite{STOC:GGSW13}, can be used to construct Identity-Based
Encryption, Attribute-Based Encryption, broadcast
encryption~\cite{TCC:Zhandry16}, and secret sharing for $\NP$ statements.  
Witness encryption is a form of encryption where a public key is simply an $\NP$ statement, and a secret key is a witness for that statement.  
More precisely, a witness encryption scheme is a pair of algorithms:
\begin{itemize}
	\item $\Enc(x,m)$ is a randomized polynomial-time algorithm that takes as input an $\NP$ statement $x$ and a message $m$, and outputs a ciphertext $c$;
    \item $\Dec(x,w,c)$ is a deterministic polynomial-time algorithm that takes as input a statement $x$, supposed witness $w$, and ciphertext $c$, and attempts to produce the message $m$.
\end{itemize}
We require that if $w$ is a valid witness for $x$, then for any message $m$, if $c\rgets\Enc(x,m)$, then $\Dec(x,w,c)$ outputs $m$  with probability 1.

The basic notion of security for witness encryption is \emph{soundness
security}, which requires that if $x$ is false, then $\Enc(x,m)$ hides
all information about $m$.  A stronger notion called \emph{extractable
security}, due to Goldwasser et al.~\cite{C:GKPVZ13}, requires, informally, that if one can learn any information about $m$ from $\Enc(x,m)$, then it must be the case that one ``knows'' a witness for $x$.

We briefly describe how to construct witness encryption from invariant maps.  
It suffices
to give a construction from any $\NP$-complete problem.  There are at
least two natural constructions from multilinear maps that we can use.
One approach is to adapt the original witness encryption scheme of Garg
et al.~\cite{STOC:GGSW13} based on the Exact Cover problem.  This approach
unfortunately also requires the same graded structure as needed by Boneh
and Waters~\cite{AC:BonWat13}.  However, we can apply the same ideas as
in our constrained PRF construction to get their scheme to work with
invariant maps.  Another is the scheme of Zhandry~\cite{TCC:Zhandry16} based on Subset Sum.\footnote{The basic scheme shown by Zhandry requires an ``asymmetric'' multilinear map, where the inputs to the map come from different sets.  However, he also explains how to instantiate the scheme using symmetric multilinear maps.  The symmetric scheme easily translates to use invariant maps.}  

As with the constructions of Garg et al.{} and Zhandry, the security of these constructions can be justified in an idealized attack model for the cryptographic invariant map, allowing only the operations explicitly allowed by the map---namely the group action and the map operation.  Justification in idealized models is not a proof, but provides heuristic evidence for security.

\section{Cryptographic invariant maps from isogenies} \label{Sec:IsogMaps}

We begin by recalling some facts that are presented in more detail in
Appendix~\ref{sec:isog}.
Let $E$ be an ordinary elliptic curve over a finite field $\F_q$ such
that the ring $\Z[\pi]$ generated by its Frobenius endomorphism $\pi$ is
integrally closed. This implies in particular that $\Z[\pi]$ is the full
endomorphism ring $\scrO$ of $E$. Let $\Cl(\scrO)$ denote the ideal class
group of this ring, and let $\Ell(\scrO)$ denote the isogeny class of $E$. There
exists a free and transitive action $\ast$ of $\Cl(\scrO)$ on
$\Ell(\scrO)$, and there is a way to represent elements of
$\Cl(\scrO)$ (namely, as products of prime ideals of small norm) that
makes this action efficiently computable. Moreover, one can 
efficiently sample close to uniform elements in $\Cl(\scrO)$ under that
representation. In addition, the ``star operator'' $\ast$ satisfies the
following property: for any choice of ideal classes
$\mathfrak{a}_1,\dots,\mathfrak{a}_n,\mathfrak{a}'_1,\dots,\mathfrak{a}'_n$
in $\Cl(\scrO)$, the abelian varieties
\begin{equation}
\label{eq:prodabvarisomcond}
(\mathfrak{a}_1\ast E)\times\cdots\times(\mathfrak{a}_n\ast E)
   \quad\textrm{and}\quad
   (\mathfrak{a}'_1\ast E)\times\cdots\times(\mathfrak{a}'_n\ast E)
\end{equation}
are isomorphic over $\F_q$ if and only if
$\mathfrak{a}_1\cdots\mathfrak{a}_n =
\mathfrak{a}'_1\cdots\mathfrak{a}'_n$ in $\Cl(\scrO)$. In particular:
\begin{equation}
\label{eq:prodabvar}
(\mathfrak{a}_1\ast E)\times\cdots\times(\mathfrak{a}_n\ast E) \cong
(\mathfrak{a}_1\cdots\mathfrak{a}_n)\ast E\times E^{n-1}.
\end{equation}

Denote by $\Ab(E)$ the set of abelian varieties over $\F_q$ 
that are a product 
of the form~\eqref{eq:prodabvar}, and assume that we can
efficiently compute an isomorphism invariant for
abelian varieties in $\Ab(E)$. In other words, assume that we have an
efficiently computable map $\isom\colon \Ab(E)\to S$ to some set $S$ that to
any tuple $E_1,\dots,E_n$ of elliptic curves isogenous to $E$ associates
an element $\isom(E_1\times\cdots\times E_n)$ of $S$ such that 
$\isom(E_1\times\cdots\times E_n) = \isom(E'_1\times\cdots\times E'_n)$
if and only if the products $E_1\times\cdots\times E_n$ and
$E'_1\times\cdots\times E'_n$ are isomorphic as abelian varieties. 
The curves $E_i$ are given for example by their $j$-invariants, and in particular,
the ideal classes $\mathfrak{a}_i$ such that $E_i \cong \mathfrak{a}_i \ast  
E$ are not supposed to be known.

Based on such an isomorphism invariant $\isom$, we
construct a cryptographic invariant map as follows. The algorithm
$\mapgen(\lambda)$ computes a sufficiently large base field $\F_q$, and an
elliptic curve $E$ over $\F_q$ such that the ring $\Z[\pi]$ generated by its
Frobenius endomorphism is integrally closed (this can be done
efficiently: see again Appendix~\ref{sec:isog}).
The algorithm then outputs the public parameters $\pp = (X,S,G,e)$ where:
\begin{itemize}
\item $X = \Ell(\scrO)$ is the isogeny class of $E$ over $\F_q$;
\item $S$ is the codomain of the isomorphism invariant $\isom$;
\item $G = \Cl(\scrO)$ is the ideal class group of $\scrO$; and
\item the map $e_n\colon X^n\to S$ is given by
$e_n(E_1,\dots,E_n) = \isom(E_1\times\cdots\times E_n)$.
\end{itemize}

The facts recalled at the beginning of this section show that $G$ acts
efficiently on $X$ freely and transitively in the sense of
Definition \ref{Def:EFTAction}, and that the properties of Definition \ref{Def:CrInvMap} are
satisfied. In particular, the invariance property follows
from~\eqref{eq:prodabvar}, and the non-degeneracy from the fact that the
abelian varieties in~\eqref{eq:prodabvarisomcond} are isomorphic
\emph{only if} the corresponding products of ideal classes coincide. Thus,
this approach does provide a cryptographic invariant map assuming $\isom$
exists.

\begin{remark}
\label{rem:nike2rmk}
In the $2$-party case, the NIKE protocol obtained from this construction 
coincides with the isogeny key exchange protocols over ordinary curves described by
Couveignes~\cite{EPRINT:Couveignes06} and
Rostovtsev--Stolbunov~\cite{EPRINT:RosSto06}. 
\end{remark}

\begin{remark}
\label{rem:isogDDH}
The existence of $\isom$ breaks the isogeny decision Diffie--Hellman
problem. Indeed, given three elliptic curves $(\mathfrak{a} \ast
E,\mathfrak{b} \ast E, \mathfrak{c} \ast E)$ isogenous to $E$, one can
check whether $\mathfrak{c}=\mathfrak{a}\mathfrak{b}$ in $\Cl(\scrO)$ by
testing whether the surfaces 
$$(\mathfrak{a} \ast E) \times (\mathfrak{b}
\ast E) \quad {\text{ and}} \quad (\mathfrak{c} \ast E) \times E$$ 
are isomorphic. This does
not prevent the construction of secure NIKE protocols (as those can be
based on the computational isogeny Diffie--Hellman problem by applying a
hash function: see Section~\ref{sec:applications}), but currently, no efficient algorithm is
known for this isogeny decision Diffie--Hellman problem. 
\end{remark}

\begin{remark}
\label{rem:nikerrmk}
For certain applications, it would be interesting to be able to hash
to the set $X = \Ell(\scrO)$, i.e., construct a random-looking curve $E'$ in the
isogeny class of $E$ without knowing an isogeny walk from $E$ to $E'$.
An equivalent problem is to construct a random-looking elliptic curve
with exactly $\#E(\F_q)$ points over $\F_q$. This seems difficult, however;
the normal way of doing so involves the CM method, which is not efficient
when the discriminant is large. 
\end{remark}

\begin{remark}
\label{rem:supersingrmk}
One can ask whether this construction extends to the supersingular
case. Over $\F_{p^2}$ with $p$ prime, the answer is clearly no, as the isogeny class of a supersingular elliptic 
curve is not endowed with a natural free and transitive group action by an abelian
group. More importantly, isomorphism classes of products of isogenous supersingular elliptic curves over $\F_q$  are
essentially trivial at least in a geometric sense. Indeed, according to a result of Deligne (see \cite[Theorem 3.5]{Shioda78}), if $E_1,\dots,E_n,E'_1,\dots,E'_n$
are all isogenous to a supersingular elliptic curve $E$, then
\[ E_1\times\cdots\times E_n \cong E'_1\times\cdots\times E'_n \quad \text{over $\overline{\F_q}$} \]
as soon as $n\geq 2$. In fact, the result holds over any extension of the base field over which all the endomorphisms of $E$ are defined, so already over $\F_{p^2}$.
However, for a supersingular elliptic curve $E$ over a prime field $\F_p$, the number of $\F_p$-iso\-mor\-phism classes of products
$E_1\times\cdots\times E_n$ with all $E_i$ isogenous to $E$ can be large. For example, this is shown
when $n=2$ in \cite[Section 5]{XYY16}. Therefore, one could conceivably obtain a 
``commutative supersingular'' version of the construction above, which would generalize the recent
$2$-party key exchange protocol CSIDH~\cite{CSIDH}, assuming that $\F_p$-isomorphism  invariants
can be computed in that setting. Since those invariants must be arithmetic rather than geometric in
nature, however, this seems even more difficult to achieve than in the ordinary case. 
\end{remark}

\section{Some natural candidate cryptographic invariant maps}
\label{sec:failure-theta-null}

In order to instantiate a cryptosystem based on the ideas in this paper, it remains to find an efficiently computable map $\isom\colon \Ab(E)\to S$ for some set $S$, as in the previous section. Below we give evidence that several natural candidates fail, either because efficiently computing them would break the cryptographic security, or because they are not in fact isomorphism invariants. 

Our primary roadblock is that while $E_1 \times \cdots \times E_n$ and $E_1' \times \cdots \times E_n'$ can be isomorphic as unpolarized abelian varieties, they are not necessarily isomorphic as polarized abelian varieties with their product polarizations.  
The first three proposals below for invariants are invariants of the isomorphism class as polarized abelian varieties, but are not invariants of the isomorphism class as unpolarized abelian varieties. We do not know a way for the different parties to choose polarizations on their product varieties in a compatible way, to produce the same invariant, without solving the elliptic curve isogeny problem.

At present, we do not know an invariant of abelian varieties in dimension $\geq 2$ that does not require choosing a polarization, with the exception of what we call the ``Deligne invariant'', described 
below.

\ssparagraph{The theta null invariant} 
One natural candidate is given by Mumford's \emph{theta nulls}, presented
in detail in Appendix~\ref{sec:algebr-theta-funct}. Unfortunately, in order to compute even a
single theta null, one must first choose a principal polarization, and
the resulting invariant does depend on this choice of polarization in a
crucial way. 
In Proposition~\ref{prop:failure-theta-nulls} below we show that, as a result, the
theta nulls do not in fact provide an isomorphism invariant as unpolarized abelian varieties.
 
\ssparagraph{Igusa invariants} 
Suppose $n = 2$ and $\End E \otimes \Q \cong \Q(\sqrt{-d})$ with $d \in \mathbb{N}$ square-free. If $d \neq 1,3,7,15$, then for $E_1$ and $E_2$ in the isogeny class of $E$, the product $E_1 \times E_2$ is the Jacobian of a genus $2$ curve $C$ (see~\cite{HayashidaNishi}). It is possible to compute such a genus $2$ curve $C$, given a suitable principal polarization on $E_1 \times E_2$. For each such $C$, one could then compute the Igusa invariants \cite{Igusa60} of $C$.
The number of genus $2$ curves $C$ such that $E_1 \times E_2$ is isomorphic to the Jacobian variety of $C$ is large (\cite{Hayashida} and \cite[Theorem 5.1]{Lange}), and unfortunately the Igusa invariants are different for different choices of $C$.  There are many principal polarizations on each element of $\Ab(E)$, and no compatible way for the different parties to choose the same one.

\ssparagraph{Invariants of Kummer surfaces} 
When $n = 2$, another approach is to consider the Kummer surface of $A = E_1 \times E_2$, which is the quotient $K = A/\{\pm1\}$. The surface $K$ itself does not depend on a polarization. But extracting an invariant from $K$, for example as in \cite[Chapter 3]{Prolegomena}, \emph{does} depend on having a projective embedding of $K$.

\ssparagraph{Deligne invariant} 
A natural candidate is an isomorphism invariant studied by Deligne \cite{Deligne69}. Suppose $A$ is an ordinary abelian variety over $k=\F_q$. 
The Serre-Tate canonical lift of $A$ to characteristic 0 produces an abelian variety over the ring of Witt vectors $W(\bar{k})$.
Fixing an embedding $\alpha$ of $W(\bar{k})$ into $\C$, we can view this lift as a complex abelian variety ${A}^{(\alpha)}$. Let $T_\alpha(A)$ denote the first integral homology group of ${A}^{(\alpha)}$. 
The Frobenius endomorphism $F$ of $A$ also lifts to characteristic 0 and defines an action of $F$ on $T_\alpha(A)$
.  
The theorem in~\cite[\S7]{Deligne69} shows that ordinary abelian varieties  $A$ and $B$ over $\F_q$ are isomorphic if and only if there is an isomorphism $T_\alpha(A) \to T_\alpha(B)$ that respects the action of $F$.

A natural candidate for a cryptographic invariant map is the map that
sends $(E_1,\ldots,E_n)$ to the isomorphism invariant 
$$T_\alpha(E_1 \times\cdots
\times E_n)=T_\alpha(E_1)\oplus \cdots\oplus T_\alpha(E_n).$$ 
Specifying the isomorphism class of $T_\alpha(E_1 \times \cdots \times E_n)$ as a $\Z[F]$-module is equivalent to specifying the action of $F$ as a $2n \times 2n$ integer matrix, unique up to conjugacy over $\mathbb{Z}$.
However, we show in Theorem \ref{Delinvthm} below that 
being able to compute $T_\alpha(E)$ for an elliptic curve $E$ in polynomial time would yield a polynomial-time algorithm to solve the elliptic curve isogeny problem of recovering $\mathfrak{a}$ given $E$ and $\mathfrak{a}\ast E$, and conversely. 

\begin{theorem}
\label{Delinvthm}
An efficient algorithm to compute Deligne invariants $T_\alpha(E)$ on an isogeny class of ordinary elliptic curves over a finite field $k$ gives an efficient algorithm to solve the elliptic curve isogeny problem in that isogeny class. Conversely, an efficient algorithm to solve the elliptic curve isogeny problem on an isogeny class of ordinary elliptic curves over $k$ yields an efficient algorithm to compute, for some embedding $\alpha:W(\bar{k}) \hookrightarrow \C$, the Deligne invariants $T_\alpha(E)$ on the isogeny class.
\end{theorem}

\begin{proof}
Suppose that $E_1$ and $E_2$ are in the isogeny class, and suppose that for $i=1,2$ we have a $\Z$-basis $\{ u_i, v_i\}$ for $T_\alpha(E_i)$ and a $2 \times 2$ integer matrix giving the action of $F$ with respect to this basis. We will efficiently compute a fractional ideal  $\mathfrak{a}$ such that $\mathfrak{a} \ast E_1 \cong E_2$.

Let $f(t)$ be the characteristic polynomial of Frobenius acting on $E_1$ or $E_2$; these are the same since $E_1$ and $E_2$ are isogenous. Let $R = \Z[t]/(f)$ and $R_\Q = R \otimes_\Z \Q$. Then $T_\alpha(E_i)$ is a rank one $R$-module, with $t$ acting as $F$. 
Compute  $a_i, b_i\in\Z$ such that $F(u_i) = a_iu_i + b_iv_i$.
 Let $\mathfrak{a}_i$ be the fractional $R$-ideal generated by
  $1$ and $(t - a_i)/b_i$.
 Compute and output $\mathfrak{a} = \mathfrak{a}_1 \mathfrak{a}_2^{-1}$.
 
We claim that $\mathfrak{a} \ast E_1 \cong E_2$. 
Define $\lambda_i: T_\alpha(E_i) \hookrightarrow R_\Q$ by sending  $w\in T_\alpha(E_i)$ to the unique  $\lambda_i(w)\in R_\Q$ such that $\lambda_i(w) \cdot u_i = w$.
Then  $\lambda_i(u_i)=1$ and $\lambda_i(v_i) = (t - a_i)/b_i$,
so the fractional ideal $\mathfrak{a}_i$ is the image of the map $\lambda_i$. 
Suppose $M$ is a positive integer such that $M\mathfrak{a}$ is an integral ideal of $R$, and 
let $h=\lambda_2^{-1} \circ M\lambda_1$.
Then $h(T_\alpha(E_1))$ is an $R$-submodule of $T_\alpha(E_2)$. 
By \cite[\S7]{Deligne69}, the map $E \mapsto T_\alpha(E)$ is a fully faithful functor, i.e., it induces a bijection
\[
  \Hom_k(E_1,E_2) \to \Hom_R(T_\alpha(E_1), T_\alpha(E_2)).
\]
Thus $h$ arises from a unique isogeny $\phi:E_1 \to E_2$. By~\cite[\S4]{Deligne69}, the kernel of $\phi$ 
is isomorphic as an $R$-module to $T_\alpha(E_2)/h(T_\alpha(E_1))$. The latter $R$-module is isomorphic to $R/M\mathfrak{a}$, and hence is exactly annihilated by $M\mathfrak{a}$. Thus $\ker(\phi) \cong E_1[M\mathfrak{a}]$, so
  $E_2 \cong E_1/E_1[M\mathfrak{a}] \cong (M\mathfrak{a}) \ast E_1.$
Since $M\mathfrak{a}$ and  $\mathfrak{a}$ are in the same ideal class, we have $E_2 \cong \mathfrak{a} \ast E_1$, as desired.
Fractional ideals can be inverted in polynomial time by \cite[Algorithm 5.3]{belabas} or \cite[\S 4.8.4]{Cohen} (see \cite[p. 21]{belabas} for the complexity).

Conversely, suppose we have an algorithm that efficiently solves the isogeny problem in the isogeny class of an ordinary elliptic curve $E_0$.
Take $R$ as above. We show below that there exists an embedding $\alpha:W(\bar{k}) \hookrightarrow \C$ such that $T_\alpha(E_0) \cong R$. 
Given $E$ isogenous to $E_0$, use the isogeny problem algorithm to compute $\mathfrak{a}$ such that $E_0 \cong \mathfrak{a} \ast E$. Output $T_\alpha(E) = \mathfrak{a}$.

It remains to show that an embedding $\alpha:W(\bar{k}) \hookrightarrow \C$ exists such that $T_\alpha(E_0) \cong R$ and $T_\alpha(E) = \mathfrak{a}$. We follow an argument in the proof of~\cite[Theorem 2.1]{Duke-Toth}. 
There exists an elliptic curve $E'$ over $\C$ with CM by $R$ for which $H_1(E',\mathbb{Z}) \cong R$ as $R$-modules. Take any embedding $\beta: W(\bar{k}) \hookrightarrow \C$. Then the complex elliptic curve  $E^{(\beta)}$ has CM by $R$, and by the theory of complex multiplication there exists $\sigma \in \mathrm{Gal}(\C/\Q)$ such that 
$E' = \sigma(E^{(\beta)}) = E^{(\sigma\circ \beta)}$. 
Let $\alpha = \sigma\circ \beta$.
By construction, $T_\alpha(E_0) = H_1(E', \mathbb{Z}) \cong R$. 
Further, by~\cite[Prop. II.1.2]{SilvermanAT}, $T_\alpha(E) \cong \mathfrak{a} \otimes_R T_\alpha(E_0) \cong \mathfrak{a}$, as claimed.
\end{proof}

\section*{Acknowledgments}

We thank the American Institute of Mathematics (AIM) for supporting a workshop on multilinear maps where the initial
seeds for this work were developed, and 
the Banff International Research Station (BIRS) where our collaboration continued.
We also thank Michiel Kosters and Yuri Zarhin. 
Boneh was partially supported by NSF, DARPA, and ONR.
Silverberg was partially supported by a grant from the Alfred P.\ Sloan Foundation and NSF grant CNS-1703321.  

\newcommand{\etalchar}[1]{$^{#1}$}

\appendix

\section{Background on isogenies}
\label{sec:isog}

This appendix provides background on isogenies of ordinary elliptic
curves over finite fields and their products, which expands upon the
short discussion at the beginning of Section~\ref{Sec:IsogMaps} and makes it explicit.
Throughout this appendix, $E$ is a fixed ordinary elliptic curve over the
finite field $\F_q$.

\subsection{Endomorphism rings} Let $\scrO = \End(E)$
be the endomorphism ring of the elliptic curve $E$. Note that since, in the ordinary case, all
endomorphisms over the algebraic closure are already defined over $\F_q$,
it is not necessary to distinguish between the endomorphisms of $E$ defined over $\F_q$ or over 
${\overline{\F_q}}$.

Denote by $\pi\in\scrO$ the Frobenius endomorphism of $E$, and by $\scrO_K$
the integral closure of $\scrO$, i.e., the ring of integers of
$\End(E)\otimes_{\Z}\Q$. We have an inclusion of orders 
$\Z[\pi] \subset \scrO \subset \scrO_K$
and both inclusions are strict in general.

However, the conductor $m=\big[\scrO_K:\Z[\pi]\big]$ is typically quite
small. Let $t$ be the trace of the Frobenius $\pi$ (computed using
Schoof's algorithm) and $D_\pi = t^2 - 4p$ the discriminant of $\Z[\pi]$;
then $m$ is given by $D_\pi = m^2 D$ where $D$ is the (fundamental)
discriminant of $\scrO_K$. Hence, $m^2$ is the quotient of $D_{\pi}$ by
its squarefree part (or $4$ times its squarefree part when $D_{\pi}$ is
even). This is usually a small number, and equal to $1$ with constant
probability for a randomly chosen curve $E$.

Therefore, we will assume for simplicity that $m=1$, and thus
$\Z[\pi] = \scrO = \scrO_K$. Note that strictly speaking, it is not known
how to test for that condition in polynomial time (squarefreeness is not
known to be easier than factoring), but we can easily test for stronger
conditions, such as $D_{\pi}$ being prime or the product of a large prime
by a small squarefree integer, and since restricting attention to such
special discriminants is desirable anyway (to avoid smooth class
numbers), this is a reasonable test to carry out.

In principle, it should be possible to extend what follows to the general
case of an arbitrary endomorphism ring $\scrO$ as long as we restrict
attention to, on the one hand, the subset $\Ell(\scrO)$ of the isogeny
class of $E$ over $\F_q$ consisting of curves with the same endomorphism
ring $\scrO$ as $E$, and on the other hand, isogenies of degree prime to
$m$. In particular, the latter are guaranteed to preserve the endomorphism
ring, and suffice to apply the expander graph argument of
Jao--Miller--Venkatesan~\cite{JaoMilVen09}; see below. There are a few
complications in considering the general case, however, so the extension
is not considered in the present discussion.

\subsection{Classification of abelian varieties isogenous to a power of
$E$}\label{sec:classif}

Let $\Ab(E)$ denote the category of abelian varieties over
$\F_q$ isogenous\footnote{The isogeny can be defined over $\F_q$ or the
algebraic closure; the two notions coincide.} to a product of copies of
$E$.
In the case when $\Z[\pi] = \scrO_K = \scrO$, the functor $\Hom(E,-)$
is an equivalence between the category $\Ab(E)$ and the category $\Mod(\scrO)$ of torsion-free $\scrO$-modules of
finite type. The inverse functor associates to any module
$M\in\Mod(\scrO)$ with finite presentation $\scrO^n\to \scrO^m\to M\to 0$
the abelian variety $M \otimes_{\scrO} E$ defined by taking the corresponding
cokernel $E^n\to E^m\to M \otimes_{\scrO} E\to 0$. This is established in
Serre's appendix~\cite{Serre02} to \cite{Lauter02}. See also \cite{Kani11}
and \cite{JKPRST18} for more general related results.

In dimension $1$, this implies that the set
$\Ell(\scrO)$ of isomorphism classes of elliptic curves over $\F_q$ isogenous to $E$ is in
bijection with the set of isomorphism classes of torsion-free modules of
rank $1$ over $\scrO$. Since any such module is isomorphic to an ideal of
$\scrO$, this yields a bijection between $\Ell(\scrO)$ and the ideal
class group $\Cl(\scrO)$ of $\scrO$.

Under that correspondence, an ideal $\mathfrak{a}$ maps to the elliptic
curve $\mathfrak{a}\otimes_{\scrO} E$ given by $E/E[\mathfrak{a}]$, where
$E[\mathfrak{a}]$ is the intersection of the kernels of all endomorphisms
in $\mathfrak{a}$ (that notation is consistent with the usual one for
torsion subgroups). Denote that elliptic curve by $\mathfrak{a}\ast E$.
Its endomorphism ring is clearly still $\scrO$, and the identification is
canonical (e.g., because the Frobenius endomorphism on each curve is
well-defined). We can therefore consider the elliptic curve
$\mathfrak{b}\ast\big(\mathfrak{a}\ast E\big)$ for some other ideal
$\mathfrak{b}$ of $\scrO$, and it turns out to be isomorphic to
$\mathfrak{ab}\ast E$. Thus, $\ast$ endows $\Ell(\scrO)$ with
a group action by $\Cl(\scrO)$, under which it is a principal homogenous
space.

By Steinitz's classification theorem of torsion-free modules of finite type over
Dedekind domains~\cite{Steinitz} (which
says that any such module is isomorphic to a sum of ideals, with equality
when the product ideals are in the same class) and the
equivalence of categories, we see that any abelian variety $A\in\Ab(E)$
is of the form $E_1\times\cdots\times E_g$ for some elliptic curves $E_i$
of the form $\mathfrak{a}_i\ast E$. Moreover, two such abelian varieties,
\begin{align*}
A &= E_1\times\cdots\times E_g \qquad (E_i = \mathfrak{a}_i\ast E) \quad \textrm{ and}\\
A'&= E'_1\times\cdots\times E'_g \qquad (E'_i = \mathfrak{a}'_i\ast E),
\end{align*}
are isomorphic if and only if $[\mathfrak{a}_1\cdots\mathfrak{a}_g] =
[\mathfrak{a}'_1\cdots\mathfrak{a}'_g]$ in $\Cl(\scrO)$. In particular,
for any fixed $g\geq 1$, the isomorphism classes of abelian varieties of
dimension $g$ are in bijection with $\Cl(\scrO)$,
given by mapping $[\mathfrak{a}]\in\Cl(\scrO)$ to $(\mathfrak{a}\ast
E)\times E^{g-1}$.

\subsection{Random sampling and isogeny walks.}\label{sec:rand-sampl-isog}
The group $\Cl(\scrO)$ is usually large; by Dirichlet's class number
formula combined with Littlewood's estimate for the values at $1$ of the
$L$-functions of quadratic characters~\cite{Littlewood28}, its order
$h(D)$ satisfies
\[ h(D) \gg \frac{\sqrt{|D|}}{\log\log |D|} \]
if one assumes that the Generalized Riemann Hypothesis holds for quadratic fields 
(where $D=t^2-4q$ can be chosen of the same order of magnitude as $q$ in
absolute value). Moreover, the best known algorithms to compute its
structure or even its order are subexponential.

We would therefore like to be able to sample uniformly from $\Cl(\scrO)$
without having to compute the group structure or the order, and do so in
such a way that if the sampling algorithm returns the class of an ideal
$\mathfrak{a}$, it is feasible to compute $\mathfrak{a}\ast E$ (which
cannot be done directly with an ideal of large norm, as it involves
computing a large degree isogeny).

To solve this problem, the same approach as in the supersingular case
basically works: simply obtain an element of the class group as a product
of sufficiently many ideals of small prime norm (and the corresponding
isogeny as the composition of the corresponding isogenies of small prime
degrees). Contrary to the supersingular case, it is not possible to use a
fixed degree, but using all small primes up to a polynomial bound is
enough, as we will discuss below.

\sssparagraph{Prime ideals and $\ell$-isogenies}
Before we discuss the sampling of random classes in $\Cl(\scrO)$, we
must recall a few facts about the ideals of the quadratic ring $\scrO$
and the corresponding isogenies. All of this is described in particular
in \cite{EPRINT:Couveignes06}, for example.
Any ideal $\mathfrak{a}$ of $\scrO$ can be factored uniquely
$\mathfrak{a} = \mathfrak{l}_1^{e_1} \cdots \mathfrak{l}_r^{e_r}$
into a product of powers of prime ideals $\mathfrak{l}_i$. To any prime
ideal $\mathfrak{l}$, one can associate the unique rational prime
$\ell\in\Z^+$ above which it sits ($\ell$ is such that
$\mathfrak{l}\cap\Z = \ell\Z$), and prime ideals can thus be obtained
from the ideal decomposition in $\scrO$ of principal ideals generated by
rational primes $\ell$.

Since $\scrO$ is quadratic, that decomposition is easy to describe in
terms of the Legendre symbol $\leg{D}{\ell}$. A rational prime $\ell$
such that $\leg{D}{\ell} = -1$ is inert in $\scrO$, i.e., $\ell\scrO$ is
 itself prime. As a result, these prime ideals are trivial in the class
group, and not interesting for our purposes. On the other hand, a
rational prime $\ell$ such that $\leg{D}{\ell} = 1$ is completely split:
we have $\ell\scrO = \mathfrak{l}\cdot\bar{\mathfrak{l}}$ for distinct
prime ideals $\mathfrak{l}\neq \bar{\mathfrak{l}}$ of norm $\ell$. This
case is the most relevant for us; it gives rise to $\ell$-isogenies.

Finally, there are finitely many rational primes $\ell$ such that
$\leg{D}{\ell} = 0$, i.e., primes dividing the discriminant. Those
primes are ramified: $\ell\scrO = \mathfrak{l}^2$ for some prime ideal of
norm $\ell$. As a result, a ramified prime $\mathfrak{l}$ is of order at
most $2$ in the class group; moreover, it is a consequence of the Chinese remainder
theorem that one can  always find an ideal
$\mathfrak{a}$ of $\scrO$ coprime to $D$ such that $[\mathfrak{l}] =
[\mathfrak{a}]$ in $\Cl(\scrO)$.  Thus, we can always ignore these primes in order to make the presentation more consistent.

It results from the above that any ideal class
$[\mathfrak{a}]\in\Cl(\scrO)$ can be written as a product
$[\mathfrak{l}_1]^{e_1}\cdots [\mathfrak{l}_r]^{e_r}$ of powers of
classes of prime ideals over split rational primes $\ell_i$ (such that
$\leg{D}{\ell_i} = 1$). Computing $\mathfrak{a}\ast E$ therefore reduces
to computing $\mathfrak{l}\ast C$ for an arbitrary prime ideal
$\mathfrak{l}$ over a split rational prime $\ell$, and $C$ an arbitrary elliptic curve isogenous to $E$. This still cannot be
done when the norm $\ell$ is large, but if it is small (say
$\polylog(q)$), it can be done efficiently.

Indeed, for such a split $\ell$, there are exactly two $\ell$-isogenies
from $C$: if we let $\ell\scrO = \mathfrak{l}\cdot\bar{\mathfrak{l}}$,
these are $C\to C/C[\mathfrak{l}]$ and $C\to C/C[\bar{\mathfrak{l}}]$. As
usual for low degree isogenies, they can be computed efficiently by
evaluating modular polynomials. Moreover, we can easily distinguish
between the two by looking at the action of the Frobenius endomorphism
$\pi$ on the $\ell$-torsion $C[\ell]$ (viewed as an $\F_\ell$-vector space
of dimension $2$).

More precisely, since $\leg{D}{\ell} = 1$, the characteristic polynomial
$X^2 - tX + q$ of the Frobenius $\pi$ has two distinct roots $\mu\neq
\bar{\mu}$ modulo $\ell$, and we can then write the ideals $\mathfrak{l},
\bar{\mathfrak{l}}$ as $(\ell,\pi-\mu)$ and $(\ell,\pi-\bar{\mu})$
respectively (up to reordering). If we view $C[\ell]$ as an
$\F_\ell$-vector space of dimension $2$, the two kernels
$C[\mathfrak{l}]$ and $C[\bar{\mathfrak{l}}]$ will be vector subspaces given by
\[ C[\mathfrak{l}] = C[\ell]\cap\Ker(\pi-\mu) \quad\text{and}\quad
   C[\bar{\mathfrak{l}}] = C[\ell]\cap\Ker(\pi-\bar{\mu})
\]
respectively, according to their generators above. In other words, these
are the eigenspaces of the Frobenius on $C[\ell]$, and $\pi$ acts by
multiplication by $\mu$ on $C[\mathfrak{l}]$ (resp.{} $\bar{\mu}$ on
$C[\bar{\mathfrak{l}}]$).

\sssparagraph{Random sampling in $\Cl(\scrO)$}
We have just seen that we can efficiently compute $\mathfrak{l}\ast C$
for any prime ideal $\mathfrak{l}$ of $\scrO$ over a split rational prime
$\ell = \polylog(q)$. We now argue that those prime ideals are in fact
sufficient to sample from a distribution close to the uniform
distribution in $\Cl(\scrO)$.

This results from a theorem of Jao, Miller and
Venkatesan~\cite[Th.~1.1]{JaoMilVen09} essentially saying that the Cayley
graph of $\Cl(\scrO)$ obtained using the classes of the ideals
$\mathfrak{l}$ associated to all split rational primes $\ell \leq x$ with
$x = \log^B |D|$ is an expander graph for any $B>2$. Moreover, there is an
explicit bound on the spectral gap; namely, if we denote by
$\lambda_{\text{triv}} \sim 2x\log x$ the largest eigenvalue of the
adjacency matrix, then all other eigenvalues $\lambda$ satisfy
\[ |\lambda| = O\big(
(\lambda_{\text{triv}}\log\lambda_{\text{triv}})^{1/2 + 1/B} \big). \]
If we take $B=2+\varepsilon$, this implies that, for a suitable absolute
constant $C$, a random walk of length
\begin{equation}
\label{eq:walklen}
r \geq C\cdot\frac{\delta + \log h(D)}{\varepsilon\cdot \log\log|D|}
\end{equation}
in the graph yields a distribution that is
$\exp(-\delta)$-statistically close to uniform.

Therefore, if we denote by $S$ the set of ideals $\mathfrak{l}$
associated to all split rational primes $\ell \leq x$ as above, and we
choose $r$ elements $\mathfrak{l}_1,\dots,\mathfrak{l}_r$ in $S$ at
random with $r$ as in~\eqref{eq:walklen}, then the class of the ideal
$\mathfrak{a} = \mathfrak{l}_1\cdots\mathfrak{l}_r$ is close to uniformly
distributed in $\Cl(\scrO)$.

We can then compute the elliptic curve $\mathfrak{a}\ast E$ as
$\mathfrak{l}_r \ast \Big(\mathfrak{l}_{r-1}\ast \big(\cdots \ast
(\mathfrak{l}_1\ast E)\big)\Big)$, and its isomorphism class is close to
uniformly distributed in the isogeny class $\Ell(\scrO)$ of $E$.%

\subsection{A low-brow approach to the classification result}
\label{sec:lowbrow}

In Appendix~\ref{sec:classif} above, we  recalled  results of
Serre and Steinitz that are sufficient to establish the following
properties, used in the main construction of Section~\ref{Sec:IsogMaps}.
Suppose that $E$ is such that $\scrO = \Z[\pi]$ is integrally closed, and
let $\mathfrak{a}_1,\dots,\mathfrak{a}_n,\mathfrak{a}'_1,\dots,
\mathfrak{a}'_n$ be ideals of $\scrO$. Then, the following abelian
varieties are isomorphic:
\begin{equation}
\label{eq:prodabvar-lb}
(\mathfrak{a}_1\ast E)\times\cdots\times(\mathfrak{a}_n\ast E) \cong
(\mathfrak{a}_1\cdots\mathfrak{a}_n)\ast E\times E^{n-1}.
\end{equation}
and more generally, we have
\begin{equation}
\label{eq:prodabvarisomcond-lb}
\begin{split}
(\mathfrak{a}_1\ast E)\times\cdots\times(\mathfrak{a}_n\ast E) \cong
(\mathfrak{a}'_1\ast E)\times\cdots\times(\mathfrak{a}'_n\ast E)
&\quad\text{if and only if}\\
\mathfrak{a}_1\cdots\mathfrak{a}_n =
\mathfrak{a}'_1\cdots\mathfrak{a}'_n
&\quad\text{as ideal classes in $\Cl(\scrO)$.}
\end{split}
\end{equation}
As a side note, we
now mention that those properties can \emph{in part} be established using
elementary techniques.
More precisely, \eqref{eq:prodabvar-lb} is a consequence of the following
elementary result.
\begin{theorem}
\label{thm:isomg}
Let $E$ be an elliptic curve over a finite field $\F_q$, and $K$ a finite
\'etale subgroup of $E$ (i.e., the map $E\to E/K$ is separable) defined over
$\F_q$. Suppose that $K$ contains subgroups $K_i$ defined over $\F_q$, for $1\leq i\leq n$, whose orders are
pairwise coprime, and suppose $K=K_1+\cdots+K_n$.
Then: 
\[ (E/K_1)\times\cdots\times (E/K_n) \cong (E/K)\times E^{n-1}. \]
\end{theorem}
\begin{proof}
The result is immediate for $n=1$. 
We next prove the result for $n=2$ by constructing an explicit
isomorphism. Consider the commutative diagram:
\[ \xymatrix{%
E \ar[r]^{\varphi_1} \ar[d]_{\varphi_2} \ar[dr]^{\theta} &
E/K_1 \ar[d]^{\psi_1} \\
E/K_2 \ar[r]_{\psi_2} & E/K}
\]
where all maps are the natural quotient isogenies. If we denote by $m_1$
and $m_2$ the orders of $K_1$ and $K_2$, we have $\deg\varphi_1 =
\deg\psi_2 = m_1$ and $\deg\varphi_2 = \deg\psi_1 = m_2$. Now choose integers
$a,b\in\Z$ such that $am_1+bm_2 = 1$. We define morphisms
\[ f\colon E\times (E/K)\to (E/K_1) \times (E/K_2) \qquad\text{and}\qquad
   g\colon (E/K_1) \times (E/K_2)\to E\times (E/K)
\]
by the following matrices:
\[ \Mat(f) = \begin{pmatrix} \varphi_1 &  \widehat{\psi_1} \\
                           -b\varphi_2 & a\widehat{\psi_2}
   \end{pmatrix} \qquad\text{and}\qquad
   \Mat(g) = \begin{pmatrix} a\widehat{\varphi_1} & -\widehat{\varphi_2} \\
                             b\psi_1              & \psi_2
   \end{pmatrix}.
\]
Then:
\begin{align*}
\Mat(g\circ f) &= 
\begin{pmatrix} a\widehat{\varphi_1} & -\widehat{\varphi_2} \\
b\psi_1 & \psi_2 \end{pmatrix} \begin{pmatrix} \varphi_1 &
\widehat{\psi_1} \\ -b\varphi_2 & a\widehat{\psi_2} \end{pmatrix} \\
&= \begin{pmatrix}
a\widehat{\varphi_1}\varphi_1 + b\widehat{\varphi_2}\varphi_2 &
a\widehat{\varphi_1}\widehat{\psi_1} - a\widehat{\varphi_2}\widehat{\psi_2}\\
b\psi_1\varphi_1 - b\psi_2\varphi_2 &
b\psi_1\widehat{\psi_1} + a\psi_2\widehat{\psi_2} \end{pmatrix} \\
&= \begin{pmatrix}
a[m_1] + b[m_2] & (a-a)\widehat{\theta} \\ (b-b)\theta & b[m_2] + a[m_1]
\end{pmatrix} = \begin{pmatrix} 1 & 0 \\ 0 & 1 \end{pmatrix}.
\end{align*}
Thus, $g\circ f$ is the identity on $E\times (E/K)$, and similarly $f\circ g$ is the identity on $(E/K_1) \times (E/K_2)$. 
So $f$ and $g$ are mutually inverse isomorphisms, proving the case $n=2$.

We now proceed by induction. Suppose the desired result holds for some $n$, and suppose
$K$ can be written as $K=K_1+\cdots+K_{n+1}$ with the $K_i$'s
of pairwise coprime order. Let $H=K_1+\cdots+K_n$. The induction hypothesis implies
\begin{equation}
\label{npart1}
(E/K_1)\times\cdots\times (E/K_n) \cong (E/H) \times E^{n-1}. 
\end{equation}
On the other hand, $K = H+K_{n+1}$ and the subgroups $H$ and
$K_{n+1}$ have coprime order, so by the $n=2$ case above:
\begin{equation}
\label{npart2}
(E/H) \times (E/K_{n+1}) \cong (E/K) \times E. 
\end{equation}
Combining \eqref{npart1} and \eqref{npart2} gives:
\begin{align*}
 (E/K_1)\times\cdots\times (E/K_n)\times (E/K_{n+1}) &\cong
 ((E/H) \times E^{n-1}) \times (E/K_{n+1}) \\ &\cong ((E/H)\times (E/K_{n+1}))\times
   E^{n-1} \cong (E/K) \times E^n
\end{align*}
which completes the induction.
\end{proof}

Applying Theorem~\ref{thm:isomg} to $K_i = E[\mathfrak{a}_i]$, we obtain \eqref{eq:prodabvar-lb} in the special case when the $K_i$'s
have pairwise coprime orders, i.e., when the ideals $\mathfrak{a}_i$ have
pairwise coprime norms. The general case follows by observing that for
any choice of the $\mathfrak{a}_i$'s, one can find ideals
$\mathfrak{b}_i$ with pairwise coprime norms such that $\mathfrak{b}_i$
is in the same ideal class as $\mathfrak{a}_i$ (just let $\mathfrak{b}_1
= \mathfrak{a}_1$, and construct $\mathfrak{b}_{i+1}$ as an ideal in the
class of $\mathfrak{a}_{i+1}$ coprime to
$N(\mathfrak{a}_1\cdots\mathfrak{a}_i)\cdot\scrO$).

As a result, the ``if'' part of property~\eqref{eq:prodabvarisomcond-lb}
also follows. Indeed, by \eqref{eq:prodabvar-lb},
$$
(\mathfrak{a}_1\ast E)\times\cdots\times(\mathfrak{a}_n\ast E) \cong (\mathfrak{a}_1\cdots\mathfrak{a}_n)\ast E\times E^{n-1}$$ and
$$(\mathfrak{a}'_1\ast E)\times\cdots\times(\mathfrak{a}'_n\ast E) \cong (\mathfrak{a}'_1\cdots\mathfrak{a}'_n)\ast E\times E^{n-1}.$$
 If the ideal classes of
$\mathfrak{a}_1\cdots\mathfrak{a}_n$ and
$\mathfrak{a}'_1\cdots\mathfrak{a}'_n$ are equal in $\Cl(\scrO)$, then 
\[ 
(\mathfrak{a}_1\ast E)\times\cdots\times(\mathfrak{a}_n\ast E) \cong
(\mathfrak{a}'_1\ast E)\times\cdots\times(\mathfrak{a}'_n\ast E)
\]
as desired.

However, the converse (the ``only if'' part of
property~\eqref{eq:prodabvarisomcond-lb}) is more difficult, and essentially states that for $E'$ and $E''$ isogenous to $E$, if $E'\times
E^{n-1}\cong E''\times E^{n-1}$ then $E'\cong E''$. Such a
cancellation result does not hold for a supersingular $E$, so the proof
in the ordinary case has to rely on the structure of the endomorphism
ring in a crucial way. See~\cite{Shioda77} for an illuminating
discussion, particularly Section~4 for the $n=2$ case.

\section{Mumford's invariant}
\label{sec:algebr-theta-funct}

In this appendix we give details of the result on Mumford's theta null invariants that was discussed in \cref{sec:failure-theta-null}.

\subsection{Theta functions}
\label{sec:theta-functions}

We review what we need from Mumford's theory of theta functions. Let $A$ be a $g$-dimensional abelian variety over an algebraically closed field $k$. In practice, we need only work over a low degree extension of our base field. We will assume $A$ comes equipped with a principal polarization given by an invertible sheaf $L_0$; the choice of such a polarization will later be crucial for our analysis. Let $L$ be an ample symmetric invertible sheaf on $A$ with $\mathrm{char}{(k)} \nmid \deg L$. Going forwards, all sheaves will be of this form. Typically we let $L = L_0^m$ for some $m$ not divisible by $\mathrm{char}(k)$.

For $P \in A$, write $\tau_P$ for the translation-by-$P$ map. Let $H(L) \subset A$ be the set of $P$ for which $\tau_P^*L \cong L.$
For convenience, we will assume $H(L) = A[m]$ for some choice of $m$. Let $G(L)$ be the \emph{theta group} of $L$, which consists of pairs $(P, \psi_P)$ where $P \in H(L)$ and $\psi_P: L \to L$ is an isomorphism that factors through $\tau_P^*L$; that is, there is some isomorphism $\tilde{\psi}_P: \tau_P^*L \to L$ such that $\psi_P = \tilde{\psi}_P \circ \tau_P^*$. The group operation is
\[
  (P,\psi_P) (Q,\psi_Q) = (P+Q,\psi_P \circ \psi_Q),
\]
which one checks is well-defined. The theta group lies in an exact sequence
\begin{equation}
  0 \longrightarrow k^* \longrightarrow G(L) \longrightarrow H(L) \longrightarrow 0\label{eq:theta-group}
\end{equation}
where the image of $\alpha \in k^*$ is $(O,\alpha)$, and the map $G(L) \to H(L)$ is $(P,\psi_P) \mapsto P$. Let $V = \Gamma(A,L)$. Then $G(L)$ has a natural action on $V$.

Let $\delta$ be the $g$-tuple $(m,m,\dots,m)$; we call $\delta$ the \emph{level} or \emph{type} of $L$. Let $K_1 = (\Z/m\Z)^g$, let $K_2 = \Hom(K_1, k^*)$, and let $H(\delta) = K_1 \times K_2$. Write $\braket{x,y}$ for the evaluation pairing on $H(\delta)$, which we extend to a symplectic form by letting $K_1 \times \{0\}$ and $\{0\} \times K_2$ be isotropic subspaces. Let $G(\delta)$, the \emph{abstract theta group of level $\delta$}, be the set $k^* \times H(\delta)$ with group operation
\[
  (a, x, y) \cdot (b, x', y') = (\braket{x,y'}ab, x+x',y+y').
\]
Then $G(\delta)$ lies in the exact sequence
\[
  0 \longrightarrow k^* \longrightarrow G(\delta) \longrightarrow H(\delta) \longrightarrow 0.
\]
Let $V(\delta)$ be the $k$-vector space of maps from $K_1$ to $k$. Then $G(\delta)$ acts on $V(\delta)$ via
\[
  (a,x,y)(f)(z) = a\braket{z,y}f(x+z).
\]
A {\bf{level $\delta$ theta structure}} for $A$ is a choice of $L$ as
above as well as a choice of isomorphism $\alpha:G(\delta) \to G(L)$ that
is the identity on $k^*$. Mumford~\cite[p.~298]{Mumford} shows that $\alpha$ induces an isomorphism $\beta:V(\delta) \to V$, unique up to a scalar, that respects $\alpha$. That is, for all $g \in G(\delta)$ and $f \in V(\delta)$, we have
$\beta(gf) = \alpha(g)\beta(f).$
Let $x_0,\dots,x_N$ be a fixed ordering of the elements of $K_1$, and let $f_i$ be the delta function at $x_i$; that is, $f_i(x_i) = 1$ and $f_i(x_j) = 0$ for $i \neq j$. The $f_i$ then furnish a canonical basis for $V(\delta)$. Let $\theta_i = \beta(f_i)$. Define $\phi: A \longrightarrow \pro^N$ by $\phi(P) = [\theta_0(P):\cdots:\theta_N(P)]$. Then $\phi(O)$ is the {\bf{theta null associated to $\alpha$}}. Thus, a choice of theta structure gives rise to a point $\phi(O)$ in projective space.

\subsection{Properties of theta structures}
\label{sec:prop-theta-struct}

\begin{proposition} \label{prop:level-4-embedding}
  Given a theta structure $\alpha: G(\delta) \to G(L)$, if $4 \mid \delta$, then the associated map $\phi:A \to \pro^N$ is an embedding. Furthermore, $\phi(A)$ is determined by the theta null $\phi(O)$ associated to $\alpha$.
\end{proposition}
Since for us, $\delta = (m, m, \dots, m)$, $4 \mid \delta$ means $4 \mid m$. The above proposition is the Corollary on p.~349 of \cite{Mumford}. The equations defining $\phi(A)$ are usually called the \emph{Riemann equations}. 
As a consequence, we see that the theta null associated to a theta structure of type $\delta$, with $4 \mid \delta$, is a nondegenerate invariant; that is, with $\delta$ chosen beforehand, nonisomorphic abelian varieties cannot have the same theta null.

\begin{lemma}
  Let $\alpha:G(\delta) \to G(L)$ be a theta structure. Then the induced map $H(\delta) \to H(L)$ is a symplectic isomorphism, where the pairing on $H(L)$ is the Weil pairing.
\end{lemma}

\begin{proof}
  See p.~317 of \cite{Mumford}.
\end{proof}

\begin{lemma}\label{lemma:symplectic-to-theta}
  Suppose $H(L) = A[m]$. Given a symplectic isomorphism $h:H(\delta) \to H(L)$, there are exactly $m^{2g}$ theta structures $\alpha$ that induce $h$.
\end{lemma}

\begin{proof}
  Let $x_1,\dots,x_g,y_1,\dots,y_g$ be a symplectic basis for $H(\delta)$. Let $P_i = h(x_i)$ and $Q_i = h(y_i)$. We have that $G(\delta)$ is generated by $k^*$ and the elements $(1,x_i,0)$ and $(1,0,y_i).$
Each of these elements has order $m$ in $G(\delta)$. Let $\alpha:G(\delta)\to G(L)$ be an isomorphism that induces $h$. Then $\alpha((1,x_i,0)) = (P_i, \psi_{P_i})$ and  $\alpha((1,0,y_i)) = (Q_i, \psi_{Q_i})$ for appropriate choices of $\psi_{P_i}$ and $\psi_{Q_i}$; in particular, each of these maps $V \to V$ must have order dividing $m$. From the exact sequence~\eqref{eq:theta-group}, we see that for fixed $i$, any two choices of $\psi_{P_i}$ (resp.~$\psi_{Q_i}$) differ by a nonzero scalar. But such a scalar must be an $m$th root of unity so that $\psi_{P_i}$ (resp.~$\psi_{Q_i}$) has order $m$. Thus, for each $P_i$, there are at most $m$ choices of $\psi_{P_i}$, and similarly for the $Q_i$. Since the choices of $\alpha((1,x_i,0))$ and $\alpha((1,0,y_i))$ completely determine $\alpha$, the claim follows.
\end{proof}

\begin{lemma}
  Let $A_1, A_2$ be abelian varieties and $L_1, L_2$ ample invertible sheaves on $A_1, A_2$, respectively. For $i=1,2$ let $p_i: A_1 \times A_2 \to A_i$ be the projection maps. Then
\[
  G(p_1^*L_1 \otimes p_2^*L_2) = (G(L_1) \times G(L_2))/\{(a,a^{-1})|a \in k^*\}.
\]
\end{lemma}

\begin{proof}
  This is Lemma 1 on p.~323 of \cite{Mumford}.
\end{proof}
We will write the right-hand group above as $G(L_1) \times_{k^*} G(L_2)$; if we think of the relevant groups as group schemes, then we do in fact get the fiber product.

Observe that if $\delta_1, \delta_2$ are the types of $L_1, L_2$, respectively, then $\delta = (\delta_1, \delta_2)$ is the type of $p_1^*L_1 \otimes p_2^*L_2$ and $G(\delta) = G(\delta_1) \times_{k^*} G(\delta_2).$ It follows that if $\alpha_i$ is a theta structure for the pair $(A_i, L_i)$, then we have a product theta structure $\alpha: G(\delta) \to G(L)$. 

\begin{proposition} \label{prop:product-thetas-segre}
  Let $A_1$ and $A_2$ be abelian varieties, let $L_1$ and $L_2$ be ample invertible sheaves on $A_1$ and $A_2$ of types $\delta_1$ and $\delta_2$, respectively, and let $L = p_1^*L_1 \otimes p_2^*L_2$, an invertible sheaf on $A = A_1 \times A_2$. Let $\alpha_1$ and  $\alpha_2$ be theta structures for $(A_1, L_1)$ and $(A_2, L_2)$, respectively, and let $\alpha$ be the product theta structure for $(A, L)$. For the theta structure $\alpha_i$, let $\phi_i: A_i \to \pro^{N_i}$ be the associated morphism, and similarly define $\phi: A \to \pro^N$. Then, up to possible reordering of coordinates, $\phi = s \circ (\phi_1, \phi_2),$ where $(\phi_1, \phi_2): A_1 \times A_2 \to \pro^{N_1} \times \pro^{N_2}$ and $s$ is the Segre map $\pro^{N_1} \times \pro^{N_2} \to \pro^N$.
\end{proposition}

\begin{proof}
  We first check that the Segre map makes sense; that is, that 
  $$N = (N_1 + 1)(N_2 + 1) - 1.$$ 
  Setting $V_i = \Gamma(A_i, L_i)$ and $V = \Gamma(A, L)$, we have $\dim V_i = N_i + 1$, $\dim V = N + 1$, but also $V = V_1 \otimes V_2$. Observe that similarly, $V(\delta) = V(\delta_1) \otimes V(\delta_2)$. Therefore the dimension counts are compatible.

  By our isomorphisms $G(\delta) \cong G(\delta_1) \times_{k^*} G(\delta_2)$ and $G(L) \cong G(L_1) \times_{k^*} G(L_2)$ and the fact that $\alpha$ is obtained as
  \[
    (\alpha_1, \alpha_2): G(\delta_1) \times_{k^*} G(\delta_2) \longrightarrow G(L_1) \times_{k^*} G(L_2),
  \]
  one sees that the induced isomorphism $\beta: V(\delta) \to V$ can instead be written
  \[
    (\beta_1,\beta_2): V(\delta_1) \otimes V(\delta_2) \longrightarrow V_1 \otimes V_2
  \]
  where the $\beta_i$ are the isomorphisms induced by the $\alpha_i$. Let $f_0, \dots, f_{N_1}$ be the standard basis for $V(\delta_1)$, and $g_0, \dots, g_{N_2}$ the standard basis for $V(\delta_2)$. Then $f_i \otimes g_j$ is a basis for $V(\delta) = V(\delta_1) \otimes V(\delta_2)$, and one sees that it will be a permutation of the standard basis. Let us assume that the standard ordering is given by 
\[
f_0 \otimes g_0, f_1 \otimes g_0, \dots, f_0 \otimes g_1, f_1 \otimes g_1, \dots.
\]
Let $r_i = \beta_1(f_i)$ and $s_j = \beta_2(g_j)$. Then for $P = (P_1,P_2) \in A_1 \times A_2$,
  \begin{align*}
    \phi_1(P_1) &= [r_0(P_1): \cdots :r_{N_1}(P_1)], \\
    \phi_2(P_2) &= [s_0(P_2): \cdots :s_{N_2}(P_2)], \text{ and} \\
    \phi(P) &= [(r_0s_0)(P): (r_1s_0)(P) : \cdots : (r_{N_1}s_{N_2})(P)] \\
           &= [r_0(P_1)s_0(P_2): \cdots : r_{N_1}(P_1)s_{N_2}(P_2)].
  \end{align*}
  Then visibly $s(\phi_1(P_1),\phi_2(P_2)) = \phi(P)$. The claim follows.
\end{proof}

Observe that under the hypotheses of the proposition, if $\phi$ is an embedding, then $\phi(A)$ is simply the Segre embedding applied to the product embedding 
$$\phi_1(A_1) \times \phi_2(A_2) \subset \pro^{N_1} \times \pro^{N_2}.$$

\subsection{Mumford's theta null invariant}
\label{sec:theta-null-invariant-appendix}

Fix $n \in\Z_{>0}$ and mutually isogenous elliptic curves  $E_1,\ldots,E_n$.
Let 
$$A = E_1 \times E_2 \times \cdots \times E_n$$ 
with its product polarization, i.e., the principal polarization with divisor
\begin{equation}
  D = (\{O\} \times E_2 \times \cdots \times E_n) + (E_1 \times \{O\} \times E_3 \times \cdots \times E_n) + \cdots + (E_1 \times \cdots \times \{O\}).\label{eq:D}
\end{equation}
Fix $m \geq 2$ and let $L$ be the invertible sheaf associated to the divisor $mD$, so that $H(L) = A[m]$ and $\delta = (m,m,\dots,m)$.
Loop over every theta structure for $L$ and, for each one, compute the associated theta null. We define $$\Mumf_m(E_1,\dots,E_n)$$ to be the resulting set of theta nulls.

\begin{proposition}\label{prop:bounded-number-theta-nulls-proven}
  For fixed $m$ and $n$, we have 
  $$\#\Mumf_m(E_1,\dots,E_n) \le m^{2n} \cdot \#\mathrm{Sp}_{2n}(\Z/m\Z).$$
\end{proposition}

\begin{proof}
  To choose a theta structure, we first choose a symplectic isomorphism $h:H(\delta) \to H(L)$. The number of such isomorphisms is the size of
$\mathrm{Sp}_{2n}(\Z/m\Z),$ the group of $2n \times 2n$ symplectic matrices with entries in $\Z/m\Z$. For each symplectic isomorphism $h$, by Lemma \ref{lemma:symplectic-to-theta} there are $m^{2n}$ theta structures lying above it. The claim now follows.
\end{proof}

\begin{proposition}\label{prop:failure-theta-nulls}
  Let $E_1$ and $E_2$ be isogenous ordinary elliptic curves over $\F_q$ with complex multiplication by $\Q(\sqrt{-d})$ with $d$ square-free. Let $h(d)$ be the class number of $\Q(\sqrt{-d})$. Suppose $4 \mid m\in \Z_{>0}$ and $h(d) > m^4 \# \mathrm{Sp}_{2n}(\Z/m\Z)$.
Then there exist $E_3$ and $E_4$ isogenous to $E_1$ such that $E_1 \times E_2 \cong E_3 \times E_4$ but $\Mumf_m(E_1,E_2) \neq \Mumf_m(E_3,E_4).$
\end{proposition}

\begin{proof}
  Let $A = E_1 \times E_2$. Let $P$ be the set of pairs of elliptic curves $(E_3, E_4)$ isogenous to $E_1$ for which $A \cong E_3 \times E_4$. Suppose for the sake of contradiction that $M:=\Mumf_m(E_1,E_2) = \Mumf_m(E_3,E_4)$ for all $(E_3, E_4) \in P$. Construct a map $F:P \to M$ as follows. Given $(E_3, E_4) \in P$, define the invertible sheaves $L_3 = L(mO)$ on $E_3$, $L_4 = L(mO)$ on $E_4$, and $L = p_3^*L_3 \otimes p_4^*L_4$ on $E_3 \times E_4$, where the $p_i$ are the obvious projections. Let $\delta=(m,m)$. Then our theta nulls for $E_3 \times E_4$ are computed using theta structures of the form $\alpha: G(\delta) \to G(L)$. Choose a product theta structure $\alpha_0$. Then define $F(E_3,E_4)$ to be the theta null associated to $\alpha_0$.

  We claim that $F$ is injective. For given a theta null in the image of $F$, it comes from a product theta structure for some $(E', E'') \in P$. Since $4 \mid m$, by the second part of Proposition \ref{prop:level-4-embedding} the associated embedding $\phi(A)$ is determined by the theta null $\phi(O)$. By Proposition \ref{prop:product-thetas-segre}, 
  we have $\phi(A) = s(\phi_1(E'), \phi_2(E''))$, where $s$ is the Segre map. As $\phi$ is an embedding (by Proposition \ref{prop:level-4-embedding}), the $\phi_i$ must also be embeddings. Therefore we can recover the pair $(E', E'')$ from the theta null. This gives a map $\mathrm{im}(F) \to P$ which is a left inverse to $F$, which shows that $F$ is injective.

Let $\scrO = \End(E_1)$. The map $\Cl(\scrO) \to P$ given by $[\mathfrak{a}] \mapsto (\mathfrak{a} \ast E_1, \mathfrak{a}^{-1} \ast E_2)$ is injective, so $\# P \geq h(d)$. By Proposition \ref{prop:bounded-number-theta-nulls-proven} and our assumption on $m$ we have 
$$\# M \leq m^4\#\mathrm{Sp}_{4}(\Z/m\Z) < h(d) \le \#P,$$ 
contradicting the injectivity of $F$. The claim follows.
\end{proof}

As stated in \Cref{sec:rand-sampl-isog}, assuming the Generalized Riemann Hypothesis, the class number $h(d)$ is roughly of size $\sqrt{q}$  as $q \to \infty$. Thus the inequality $h(d) > m^4 \#\mathrm{Sp}_{2n}(\Z/m\Z)$
in the statement of Proposition \ref{prop:failure-theta-nulls} holds for typical use-cases.


\begin{thebibliography}{GGH{\etalchar{+}}13b}

\bibitem[Bel04]{belabas}
Karim Belabas.
\newblock Topics in computational algebraic number theory.
\newblock {\em J. Th\'eor. Nombres Bordeaux}, 16(1):19--63, 2004.

\bibitem[BGI14]{PKC:BoyGolIva14}
Elette Boyle, Shafi Goldwasser, and Ioana Ivan.
\newblock Functional signatures and pseudorandom functions.
\newblock In Hugo Krawczyk, editor, {\em PKC~2014}, volume 8383 of {\em
  {LNCS}}, pages 501--519. Springer, Heidelberg, March 2014.

\bibitem[BLR{\etalchar{+}}15]{EC:BLRSZZ15}
Dan Boneh, Kevin Lewi, Mariana Raykova, Amit Sahai, Mark Zhandry, and Joe
  Zimmerman.
\newblock Semantically secure order-revealing encryption: Multi-input
  functional encryption without obfuscation.
\newblock In Elisabeth Oswald and Marc Fischlin, editors, {\em EUROCRYPT~2015,
  Part~II}, volume 9057 of {\em {LNCS}}, pages 563--594. Springer, Heidelberg,
  April 2015.

\bibitem[BLS01]{AC:BonLynSha01}
Dan Boneh, Ben Lynn, and Hovav Shacham.
\newblock Short signatures from the {Weil} pairing.
\newblock In Colin Boyd, editor, {\em ASIACRYPT~2001}, volume 2248 of {\em
  {LNCS}}, pages 514--532. Springer, Heidelberg, December 2001.

\bibitem[BR96]{EC:BelRog96}
Mihir Bellare and Phillip Rogaway.
\newblock The exact security of digital signatures: How to sign with {RSA} and
  {Rabin}.
\newblock In Ueli~M. Maurer, editor, {\em EUROCRYPT'96}, volume 1070 of {\em
  {LNCS}}, pages 399--416. Springer, Heidelberg, May 1996.

\bibitem[BS03]{boneh2003applications}
Dan Boneh and Alice Silverberg.
\newblock Applications of multilinear forms to cryptography.
\newblock {\em Contemporary Mathematics}, 324(1):71--90, 2003.

\bibitem[BW13]{AC:BonWat13}
Dan Boneh and Brent Waters.
\newblock Constrained pseudorandom functions and their applications.
\newblock In Kazue Sako and Palash Sarkar, editors, {\em ASIACRYPT~2013,
  Part~II}, volume 8270 of {\em {LNCS}}, pages 280--300. Springer, Heidelberg,
  December 2013.

\bibitem[BZ17]{boneh2017multiparty}
Dan Boneh and Mark Zhandry.
\newblock Multiparty key exchange, efficient traitor tracing, and more from
  indistinguishability obfuscation.
\newblock {\em Algorithmica}, 79(4):1233--1285, 2017.
\newblock Extended abstract in Crypto 2014.

\bibitem[CF96]{Prolegomena}
John W.~S. Cassels and E.~Victor Flynn.
\newblock {\em Prolegomena to a middlebrow arithmetic of curves of genus
  {$2$}}, volume 230 of {\em London Mathematical Society Lecture Note Series}.
\newblock Cambridge University Press, Cambridge, 1996.

\bibitem[CLG09]{JC:ChaLauGor09}
Denis~Xavier Charles, Kristin~E. Lauter, and Eyal~Z. Goren.
\newblock Cryptographic hash functions from expander graphs.
\newblock {\em Journal of Cryptology}, 22(1):93--113, January 2009.

\bibitem[CLM{\etalchar{+}}18]{CSIDH}
Wouter Castryck, Tanja Lange, Chloe Martindale, Lorenz Panny, and Joost Renes.
\newblock {CSIDH}: An efficient post-quantum commutative group action.
\newblock Cryptology ePrint Archive, Report 2018/383, 2018.
\newblock \url{https://eprint.iacr.org/2018/383}.

\bibitem[Coh93]{Cohen}
Henri Cohen.
\newblock {\em A course in computational algebraic number theory}, volume 138
  of {\em Graduate Texts in Mathematics}.
\newblock Springer-Verlag, Berlin, 1993.

\bibitem[Cou06]{EPRINT:Couveignes06}
Jean-Marc Couveignes.
\newblock Hard homogeneous spaces.
\newblock Cryptology ePrint Archive, Report 2006/291, 2006.
\newblock \url{http://eprint.iacr.org/2006/291}.

\bibitem[Del69]{Deligne69}
Pierre Deligne.
\newblock Vari{\'e}t{\'e}s ab{\'e}liennes ordinaires sur un corps fini.
\newblock {\em Invent. Math.}, 8:238--243, 1969.

\bibitem[Deu41]{Deuring}
Max Deuring.
\newblock Die {T}ypen der {M}ultiplikatorenringe elliptischer
  {F}unktionenk\"orper.
\newblock {\em Abh. Math. Sem. Hansischen Univ.}, 14:197--272, 1941.

\bibitem[DT02]{Duke-Toth}
W.~Duke and \'A. T\'oth.
\newblock The splitting of primes in division fields of elliptic curves.
\newblock {\em Experiment. Math.}, 11(4):555--565 (2003), 2002.

\bibitem[DY05]{PKC:DodYam05}
Yevgeniy Dodis and Aleksandr Yampolskiy.
\newblock A verifiable random function with short proofs and keys.
\newblock In Serge Vaudenay, editor, {\em PKC~2005}, volume 3386 of {\em
  {LNCS}}, pages 416--431. Springer, Heidelberg, January 2005.

\bibitem[EHL{\etalchar{+}}18]{EC:EHLMP18}
Kirsten Eisentr{\"a}ger, Sean Hallgren, Kristin~E. Lauter, Travis Morrison, and
  Christophe Petit.
\newblock Supersingular isogeny graphs and endomorphism rings: Reductions and
  solutions.
\newblock In Jesper~Buus Nielsen and Vincent Rijmen, editors, {\em
  EUROCRYPT~2018, Part~III}, volume 10822 of {\em {LNCS}}, pages 329--368.
  Springer, Heidelberg, April~/~May 2018.

\bibitem[FHKP13]{PKC:FHKP13}
Eduarda S.~V. Freire, Dennis Hofheinz, Eike Kiltz, and Kenneth~G. Paterson.
\newblock Non-interactive key exchange.
\newblock In Kaoru Kurosawa and Goichiro Hanaoka, editors, {\em PKC~2013},
  volume 7778 of {\em {LNCS}}, pages 254--271. Springer, Heidelberg,
  February~/~March 2013.

\bibitem[FJP14]{FJP14}
Luca~De Feo, David Jao, and J{\'{e}}r{\^{o}}me Pl{\^{u}}t.
\newblock Towards quantum-resistant cryptosystems from supersingular elliptic
  curve isogenies.
\newblock {\em J. Mathematical Cryptology}, 8(3):209--247, 2014.

\bibitem[FKS18]{cryptoeprint:2018:485}
Luca~De Feo, Jean Kieffer, and Benjamin Smith.
\newblock Towards practical key exchange from ordinary isogeny graphs.
\newblock Cryptology ePrint Archive, Report 2018/485, 2018.
\newblock \url{https://eprint.iacr.org/2018/485}.

\bibitem[Fuc14]{SCN:Fuchsbauer14}
Georg Fuchsbauer.
\newblock Constrained verifiable random functions.
\newblock In Michel Abdalla and Roberto~De Prisco, editors, {\em SCN 14},
  volume 8642 of {\em {LNCS}}, pages 95--114. Springer, Heidelberg, September
  2014.

\bibitem[Gal18]{Galbraith}
Steven~D. Galbraith.
\newblock Authenticated key exchange for {SIDH}.
\newblock Cryptology ePrint Archive, Report 2018/266, 2018.
\newblock \url{https://eprint.iacr.org/2018/266}.

\bibitem[GGH13a]{EC:GarGenHal13}
Sanjam Garg, Craig Gentry, and Shai Halevi.
\newblock Candidate multilinear maps from ideal lattices.
\newblock In Thomas Johansson and Phong~Q. Nguyen, editors, {\em
  EUROCRYPT~2013}, volume 7881 of {\em {LNCS}}, pages 1--17. Springer,
  Heidelberg, May 2013.

\bibitem[GGH{\etalchar{+}}13b]{FOCS:GGHRSW13}
Sanjam Garg, Craig Gentry, Shai Halevi, Mariana Raykova, Amit Sahai, and Brent
  Waters.
\newblock Candidate indistinguishability obfuscation and functional encryption
  for all circuits.
\newblock In {\em 54th FOCS}, pages 40--49. {IEEE} Computer Society Press,
  October 2013.

\bibitem[GGSW13]{STOC:GGSW13}
Sanjam Garg, Craig Gentry, Amit Sahai, and Brent Waters.
\newblock Witness encryption and its applications.
\newblock In Dan Boneh, Tim Roughgarden, and Joan Feigenbaum, editors, {\em
  45th ACM STOC}, pages 467--476. {ACM} Press, June 2013.

\bibitem[GKP{\etalchar{+}}13]{C:GKPVZ13}
Shafi Goldwasser, Yael~Tauman Kalai, Raluca~A. Popa, Vinod Vaikuntanathan, and
  Nickolai Zeldovich.
\newblock How to run {T}uring machines on encrypted data.
\newblock In Ran Canetti and Juan~A. Garay, editors, {\em CRYPTO~2013,
  Part~II}, volume 8043 of {\em {LNCS}}, pages 536--553. Springer, Heidelberg,
  August 2013.

\bibitem[Hay65]{Hayashida}
Tsuyoshi Hayashida.
\newblock A class number associated with a product of two elliptic curves.
\newblock {\em Natur. Sci. Rep. Ochanomizu Univ.}, 16:9--19, 1965.

\bibitem[HN65]{HayashidaNishi}
Tsuyoshi Hayashida and Mieo Nishi.
\newblock Existence of curves of genus two on a product of two elliptic curves.
\newblock {\em J. Math. Soc. Japan}, 17:1--16, 1965.

\bibitem[Igu60]{Igusa60}
Jun-ichi Igusa.
\newblock Arithmetic variety of moduli for genus two.
\newblock {\em Ann. of Math. (2)}, 72:612--649, 1960.

\bibitem[JKP{\etalchar{+}}18]{JKPRST18}
Bruce~W. Jordan, Allan~G. Keeton, Bjorn Poonen, Eric~M. Rains, Nicholas
  Shepherd-Barron, and John~T. Tate.
\newblock Abelian varieties isogenous to a power of an elliptic curve.
\newblock {\em Compositio Math.}, 154(5):934--959, 2018.
\newblock \url{https://arxiv.org/abs/1602.06237}.

\bibitem[JMV09]{JaoMilVen09}
David Jao, Stephen~D. Miller, and Ramarathnam Venkatesan.
\newblock Expander graphs based on {GRH} with an application to elliptic curve
  cryptography.
\newblock {\em J. Number Theory}, 129(6):1491--1504, 2009.
\newblock \url{https://arxiv.org/abs/0811.0647}.

\bibitem[Jou04]{JC:Joux04}
Antoine Joux.
\newblock A one round protocol for tripartite {Diffie}-{Hellman}.
\newblock {\em Journal of Cryptology}, 17(4):263--276, September 2004.

\bibitem[Kan11]{Kani11}
Ernst Kani.
\newblock Products of {CM} elliptic curves.
\newblock {\em Collectanea Math.}, 62(3):297--339, 2011.
\newblock \url{http://www.mast.queensu.ca/~kani/papers/CMprod3.pdf}.

\bibitem[KPTZ13]{CCS:KPTZ13}
Aggelos Kiayias, Stavros Papadopoulos, Nikos Triandopoulos, and Thomas
  Zacharias.
\newblock Delegatable pseudorandom functions and applications.
\newblock In Ahmad-Reza Sadeghi, Virgil~D. Gligor, and Moti Yung, editors, {\em
  ACM CCS 13}, pages 669--684. {ACM} Press, November 2013.

\bibitem[Lan06]{Lange}
Herbert Lange.
\newblock Principal polarizations on products of elliptic curves.
\newblock In {\em The geometry of {R}iemann surfaces and abelian varieties},
  volume 397 of {\em Contemp. Math.}, pages 153--162. Amer. Math. Soc.,
  Providence, RI, 2006.

\bibitem[Lau02]{Lauter02}
Kristin Lauter.
\newblock The maximum or minimum number of rational points on genus three
  curves over finite fields.
\newblock {\em Compositio Math.}, 134(1):87--111, 2002.
\newblock \url{https://arxiv.org/abs/math/0104086}.

\bibitem[Lit28]{Littlewood28}
John~E. Littlewood.
\newblock On the class number of the corpus {$P(\sqrt{-k})$}.
\newblock {\em Proc. London Math. Soc.}, 27(1):358--372, 1928.
\newblock \url{https://doi.org/10.1112/plms/s2-27.1.358}.

\bibitem[Lys02]{C:Lysyanskaya02}
Anna Lysyanskaya.
\newblock Unique signatures and verifiable random functions from the {DH-DDH}
  separation.
\newblock In Moti Yung, editor, {\em CRYPTO~2002}, volume 2442 of {\em {LNCS}},
  pages 597--612. Springer, Heidelberg, August 2002.

\bibitem[MRV99]{FOCS:MicRabVad99}
Silvio Micali, Michael~O. Rabin, and Salil~P. Vadhan.
\newblock Verifiable random functions.
\newblock In {\em 40th FOCS}, pages 120--130. {IEEE} Computer Society Press,
  October 1999.

\bibitem[Mum66]{Mumford}
David Mumford.
\newblock On the equations defining abelian varieties. {I}.
\newblock {\em Invent. Math.}, 1:287--354, 1966.

\bibitem[Rao14]{EPRINT:Rao14}
Vanishree Rao.
\newblock Adaptive multiparty non-interactive key exchange without setup in the
  standard model.
\newblock Cryptology ePrint Archive, Report 2014/910, 2014.
\newblock \url{http://eprint.iacr.org/2014/910}.

\bibitem[RS06]{EPRINT:RosSto06}
Alexander Rostovtsev and Anton Stolbunov.
\newblock Public-key cryptosystem based on isogenies.
\newblock Cryptology ePrint Archive, Report 2006/145, 2006.
\newblock \url{http://eprint.iacr.org/2006/145}.

\bibitem[Ser02]{Serre02}
Jean-Pierre Serre.
\newblock Modules hermitiens et courbes alg{\'e}briques.
\newblock Appendix to {\cite{Lauter02}}, 2002.

\bibitem[Shi77]{Shioda77}
Tetsuji Shioda.
\newblock Some remarks on abelian varieties.
\newblock {\em J. Fac. Sci. Univ. Tokyo, Sec. 1A}, 24(1):11--21, 1977.
\newblock
  \url{http://repository-old.dl.itc.u-tokyo.ac.jp/dspace/handle/2261/6164}.

\bibitem[Shi78]{Shioda78}
Tetsuji Shioda.
\newblock Supersingular {K3} surfaces.
\newblock In Knud L{\o}nsted, editor, {\em Algebraic Geometry}, volume 732 of
  {\em Lecture Notes in Mathematics}, pages 564--591. Springer, 1978.

\bibitem[Sil94]{SilvermanAT}
Joseph~H. Silverman.
\newblock {\em Advanced topics in the arithmetic of elliptic curves}, volume
  151 of {\em Graduate Texts in Mathematics}.
\newblock Springer-Verlag, New York, 1994.

\bibitem[Ste11]{Steinitz}
Ernst Steinitz.
\newblock Rechteckige {S}ysteme und {M}oduln in algebraischen {Z}ahlk\"oppern.
  {I}.
\newblock {\em Math. Ann.}, 71(3):328--354, 1911.

\bibitem[XYY16]{XYY16}
Jiangwei Xue, Tse-Chung Yang, and Chia-Fu Yu.
\newblock On superspecial abelian surfaces over finite fields.
\newblock {\em Documenta Mathematica}, 21:1607--1643, 2016.

\bibitem[Zha16]{TCC:Zhandry16}
Mark Zhandry.
\newblock How to avoid obfuscation using witness {PRFs}.
\newblock In Eyal Kushilevitz and Tal Malkin, editors, {\em TCC~2016-A,
  Part~II}, volume 9563 of {\em {LNCS}}, pages 421--448. Springer, Heidelberg,
  January 2016.

\end{thebibliography}
\end{document}